\documentclass{amsart}

\pdfoutput=1

\usepackage[english]{babel}

\usepackage{array}
\usepackage{tabularx}

\usepackage{booktabs}
\usepackage{multirow}

\usepackage{comment}
\usepackage{graphicx}
\usepackage{geometry}
\usepackage{amsmath,amsthm,amssymb}
\usepackage{cancel}

\usepackage{thmtools}
\usepackage{thm-restate}
\usepackage{xcolor}

\usepackage{url}
\usepackage{tikz}
\usepackage{lmodern}
\usepackage{listings}
\usepackage{microtype}
\usepackage{pdfpages}
\usepackage{algpseudocode}
\usepackage{framed}
\usepackage{enumerate}
\usepackage{bm}

\usepackage{mathabx}

\newtheorem{definition}{Definition}

\newtheorem{lemma}[definition]{Lemma}
\newtheorem{theorem}[definition]{Theorem}

\newtheorem{corollary}[definition]{Corollary}

\newtheorem{conjecture}[definition]{Conjecture}

\usepackage{algorithm2e}

\newcommand{\ppol}{\ensuremath{{\mathrm{pPol}}}}
\newcommand{\pPol}{\ensuremath{{\mathrm{pPol}}}}
\newcommand{\pol}{\ensuremath{{\mathrm{Pol}}}}
\newcommand{\ar}{\ensuremath{{\mathrm{ar}}}}

\newcommand{\qfppclone}[1]{\ensuremath{\langle #1 \rangle_{\cancel{\exists}}}}
\newcommand{\cclone}[1]{\ensuremath{\langle #1 \rangle}}

\newcommand{\NP}{\textsf{NP}}
\newcommand{\Csp}{\textsc{Csp}}

\newcommand{\N}{\mathbb{N}}

\definecolor{DarkBlue}{rgb}{ 0 , 0 , 0.7}

\newtheorem{remark}[theorem]{Remark}
\newtheorem{question}[theorem]{Question}
\newtheorem*{problem1}{\sc $H$-Coloring Problem}

\begin{document}

\title[The Fine-Grained Complexity of Graph Homomorphism Problems]{The Fine-Grained Complexity of Graph Homomorphism Problems: Towards the Okrasa and Rz\c ażewski Conjecture}

\author[A. Baril]{Ambroise Baril}
\address[Ambroise Baril]{Université de Lorraine, CNRS, LORIA, France}
\email{ambroise.baril@loria.fr}

\author[M. Couceiro]{Miguel Couceiro}
\address[Miguel Couceiro]{Université de Lorraine, CNRS, LORIA, France\\INESC-ID, Instituto Superior T\'{e}cnico, Universidade de Lisboa, Lisboa, Portugal}
\email{miguel.couceiro@loria.fr}

\author[V. Lagerkvist]{Victor Lagerkvist}
\address[Victor Lagerkvist]{Department of Computer and Information Science, Link\"{o}pings Universitet, Sweden}
\email{victor.lagerkvist@liu.se}

\begin{abstract}

In this paper we are interested in the fine-grained complexity of deciding whether there is a homomorphism from an input graph $G$ to a fixed graph $H$ (the \textsc{$H$-Coloring} problem).  The starting point is that these problems can be viewed as constraint satisfaction problems (\Csp s), and that (partial) polymorphisms of binary relations are of paramount importance in the study of complexity classes of such \Csp s.

Thus, we first investigate the expressivity of binary symmetric relations $E_H$ and their corresponding (partial) polymorphisms $\pPol(E_H)$. For irreflexive graphs we observe that there is no pair of  graphs $H$ and $H'$ such that $\pPol(E_H)\subseteq \pPol(E_{H'})$, unless $E_{H'}=\emptyset$ or $H$ =$H'$. More generally we show the existence of an $n$-ary relation $R$ whose partial polymorphisms strictly subsume those of $H$ and such that \Csp$(R)$ is \NP-complete if and only if $H$ contains an odd cycle of length at most $n$.
Motivated by this we also describe the sets of total polymorphisms of nontrivial cliques,  odd cycles, as well as certain cores, and we give an algebraic characterization of {\em projective cores}.
As a by-product, we settle the {\em Okrasa and Rz\c ażewski conjecture} for all graphs of at most 7 vertices. 
\end{abstract}


\maketitle

\section{Introduction}
This paper aims to improve our understanding of fine-grained complexity of constraint satisfaction problems (\Csp s) \cite{jonsson2017}. In a {\em constraint satisfaction problem} (\Csp), given a set of variables $X$ and a set of constraints of the form $R(\mathbf{x})$ for ${\mathbf x} \in X^k$ and some $k$-ary relation $R$, the objective is to assign
values from $X$ to a domain $V$ such that every constraint in $C$ is satisfied. 
This problem is usually denoted by \Csp$(\Gamma)$, with the additional stipulation that every relation occurring in a constraint comes from the set of relations $\Gamma$, and it is typically phrased as the decision problem of verifying whether a solution exists.

In this article we take a particular interest in the restricted case when $\Gamma$ is singleton and contains a binary, symmetric relation $H$, viewed as the edge relation of a graph\footnote{Throughout, we assume that
all graphs are finite, simple (no loops) and  undirected. Every graph $H$ is defined by its set $V_H$ of vertices and its set $E_H$ of edges.}, the \textsc{$H$-Coloring} problem~\cite{hell1990complexity}. 
This problem is arguably more naturally formulated as a {\em homomorphism} problem. Recall that a function $f \colon V_G\rightarrow V_H$ is said to be a homomorphism between the two graphs $G$ and $H$ if it  ``preserves'' the edge relation, that is, if for every edge $ (u,v)\in E_G$, we have $(f(u),f(v))\in E_H$. In that case, we use the notation $f:G\rightarrow H$. For each graph $H$ the \textsc{$H$-coloring} problem can then succinctly be defined as follows.

\begin{problem1}
Given a graph $G$, decide whether there is a homomorphism $f:G\rightarrow H$.
\end{problem1}

Note that \textsc{$H$-coloring} is the same problem as $\Csp(\{E_H\})$ (henceforth written $\Csp(E_H)$).
Clearly, the \textsc{$H$-Coloring} problem subsumes the well-known $k$-\textsc{Coloring} problem, $k\geq 1$, that asks for a coloring of the vertices of a graph using at most $k$ colors such that each pair of adjacent vertices are assigned different colors. Indeed, it corresponds to the case where $H=K_k$, the complete graph (clique) of size $k$. 

Hell and  Ne$\check{\text{s}}$et$\check{\text{r}}$il \cite{hell1990complexity} showed that the \textsc{$H$-Coloring} problem is in $\textsf{P}$ (the class of problems decidable in polynomial time) whenever $H$ is bipartite, and it is  \NP-complete otherwise.
Our goal is to bring some light into the (presumably) superpolynomial complexity of the \textsc{$H$-Coloring} problem when $H$ is non-bipartite. 
    On the one hand, there are already some strong upper-bounds results on the fine-grained complexity of $k$-\textsc{Coloring} for $k\geq 3$. Björklund {\it et al.} \cite{BjorklundHK09} proved that $k$-\textsc{Coloring} is solvable in time $O^*(2^n)$ (i.e., $O(2^n\times n^{O(1)})$) where $n$ is the number of vertices of the input graph. Fomin {\it et al.} \cite{fomin2007exact} also prove that $C_{2k+1}$-\textsc{Coloring} is solvable in time $O^*( \binom{n}{n/k} )=O^*((\alpha_k)^n)$, with $\alpha_k \underset{k\rightarrow \infty}{\longrightarrow} 1$, and improved algorithms are also known when $H$ has bounded {\em tree-width} or {\em clique-width}~\cite{fomin2007exact,DBLP:journals/mst/Wahlstrom11}.
On the other hand, lower bounds by Fomin {\it et al.} \cite{fomin2015} rule out the existence of a uniform $2^{O(n)}$ time algorithm under the {\em exponential-time hypothesis} ({\it i.e.}, that 3-\textsc{Sat} can not be solved in subexponential time).

We notice, however, that there is a lack of general tools for describing fine-grained properties of \Csp s, and in particular we lack techniques for comparing \NP-hard \textsc{$H$-Coloring} problems with each other, e.g., via size-preserving reductions. We explore these ideas through an algebraic approach, by investigating algebraic invariants of graphs.
For this purpose, viewing \textsc{$H$-Coloring} as \Csp($E_H$) is quite useful as it allows us to use the widely studied theory of the complexity of \Csp($\Gamma$), since the former is just the special case when $\Gamma=\{ R\}$ is a singleton containing a binary, symmetric relation. In particular, it was shown that the fine-grained complexity of \Csp($R$)
only depends on the so-called {\em partial polymorphisms} of $R$~\cite{ismvlCHL19,jonsson2017}. Briefly, a {\em polymorphism} is a higher-arity homomorphism from the relation to the relation itself.
Additionally, a polymorphism that is not necessarily everywhere defined  is known as a {\em partial polymorphism}, and we write $\pol( R)$ (respectively, $\pPol( R)$) for the set of all (partial) polymorphisms of a relation $ R$. If $H$ is a graph, then by $\pPol(H)$, we simply mean the set of partial polymorphisms of the edge relation $E_H$. It is then known that partial polymorphisms correlate to fine-grained complexity in the sense that if $\pPol( R) \subseteq \pPol( R')$ and if \Csp$( R)$ is solvable in $O^*(c^n)$ time for some $c > 1$ then \Csp$( R')$ is also solvable in $O^*(c^n)$ time~\cite{jonsson2017}.

Thus, describing the inclusion structure between sets of the form $\pPol(H)$ would allow us to relate the fine-grained complexity of \textsc{$H$-Coloring} problems with each other, but, curiously, 
we manage to prove that {\em no} non-trivial inclusions of this form exist, suggesting that partial polymorphisms of graphs are not easy to relate via set inclusion. %
As a follow-up question we also study inclusions  of the form $\pPol(H) \subseteq \pPol(R)$ when $R$ is an arbitrary relation, and manage to give a non-trivial condition based on the length of the shortest odd cycle of $H$.
Concretely, we prove that it is possible to find an $n$-ary relation $R$ with $\pPol(E_H)\subsetneq \pPol(R)$ where  \Csp($R$) is \NP-complete, if and only if $H$ contains an odd-cycle of length at most $n$.
This result suggests that the size of the smallest odd-cycle is an interesting parameter when regarding the complexity of \textsc{$H$-Coloring}. As observed above, the smaller $\pPol(E_H)$ is, the harder \Csp($E_H$) (and thus \textsc{$H$-Coloring}) is. In other words, the greater the smallest odd-cycle of $H$ is, the easier the \textsc{$H$-Coloring} problem is. This fact supports the already known algorithms presented in \cite{fomin2007exact}.

Despite this trivial inclusion structure, it could still be of great interest to provide a succinct description of $\pPol(H)$ for some noteworthy choices of non-bipartite, core $H$. As a first step in this project we concentrate on the total polymorphisms of $H$, and conclude that {\em projective graphs}~\cite{larose2001strongly} appear to be a reasonable class to target since the total polymorphisms of projective cores are {\em essentially at most unary}. %
Projective cores were studied by Okrasa and Rz\c ażewski~\cite{okrasa2020finegrained} in the context of fine-grained aspects of \textsc{$H$-coloring} problems analyzed under tree-width. We have the following conjecture.

\begin{restatable}[\cite{okrasa2020finegrained}]{conjecture1}{Okrasa}
\label{conjecture}
Let $H$ be a connected non-trivial core on at least 3 vertices. Then $H$ is projective if and only if it is indecomposable.
\end{restatable}

Thus, we should not hope to easily give a complete description of projective cores, but we do succeed in  (1) proving that several well-known families of graphs, e.g., cliques, odd-cycles, and other core graphs, are projective, and (2) confirm the conjecture for {\em all} graphs with at most 7 vertices. Importantly, our proofs use the algebraic approach and are significantly simpler than existing proofs, and suggest that the algebraic approach might be a cornerstone in completely describing projective cores.

This paper is organized as follows. In Section~\ref{sec:prel}, we recall the basic notions and preliminary results needed throughout the paper.  We investigate the order structure of classes of graph (partial) polymorphisms in Section~\ref{sec:classification} where we show the aforementioned main results. In Section~\ref{section:Projective graphs} we focus on projective and core graphs and present several general examples of projective cores, and  settle the Okrasa and Rz\c ażewski conjecture for graphs with at most 7 vertices. %
In Section~\ref{sec:conclusion} we discuss some consequences of our results and state a few noteworthy conjectures. 

\section{Preliminaries} \label{sec:prel}

Throughout the article we use the following notation.

For any $n\in\N$, $[n]$ denotes the set $\{1,\ldots,n\}$. For every set $V$, $n\geq 1$ and $t=(t_1,\ldots,t_n)\in V^n$, $t[i]$ denotes $t_i$, and given a relation $R \subseteq [k]^n$ for some $k \geq 1$, we write $\ar(R)$ for its arity $n$. For all $m\geq 1$ and $i\in [m]$, we write $\pi_i^m:V^m\rightarrow V$ for the {\em projection} on the $i$-th coordinate (the set $V$ will always be implicit in the context).

For $H$ a graph and $V\subseteq V_H$, we denote by $H[V]$ the graph induced by $V$ on $H$.
We use the symbol $\uplus$ to express the disjoint union of sets, and $+$ to express the disjoint union of graphs.

For a unary function $f:V\rightarrow V$, and an $m$-ary function $g:V^m\rightarrow V$ we write $f\circ g$ for their {\em composition} that is defined by $(f\circ g)(x_1,\ldots,x_m)=f(g((x_1,\ldots,x_m)))$, for every $(x_1,\ldots,x_m)\in V^m$.

Also, for $k\geq 3$, $K_k$ and $C_k$ denote respectively a $k$-clique and a $k$-cycle.

\subsection{Graph homomorphisms and cores}
For two graphs $G$ and $H$ a function $f\colon V_G\to V_H$ is a {\em homomorphism} from $G$ to $H$ if $\forall (u,v)\in E_G, (f(u),f(v))\in E_H$. In this case, $f$ is also called an {\it $H$-coloring} of $G$, and we denote this fact by $f\colon G\rightarrow H$. The graph $G$ is said to be $H$-{\it colorable}, which we denote by $G\rightarrow H$, if there exists $f\colon G\rightarrow H$. For a graph $H$, the \textsc{$H$-Coloring} problem thus asks whether a given graph $G$ is $H$-colorable.

\begin{theorem}[\cite{hell1990complexity}]\label{thm:Complexity H-COLORING} \textsc{$H$-Coloring} is in $\textsf{P}$ whenever $H$ is bipartite, and is \NP-complete, otherwise.
\end{theorem}

A key notion in the proof of Theorem~\ref{thm:Complexity H-COLORING} is the notion of a {\em graph core}: let $core(H)$ be the smallest induced subgraph $H'$ of $H$  such that $H\rightarrow H'$. The graph $H$ is said to be a {\em core} if $H=core(H)$. Note that the core of a graph $H$ is unique up to isomorphism and that the problems \textsc{$H$-Coloring} and $core(H)$-\textsc{Coloring} are equivalent. Thus, for both classical and fine-grained complexity, it is sufficient to consider $core(H)$-\textsc{Coloring}. Moreover, it is not difficult to verify that cliques and odd-cycles are cores. Notice that a graph $H$ is a core if and only if every $H$-coloring of $H$ is bijective.

For two graphs $G$ and $H$, we define their {\em cross product} $G\times H$ as the graph with  $V_{G\times H}=V_G\times V_H$ and   
$$E_{G\times H}=\{ ((u_1,v_1),(u_2,v_2)) ~|~ (u_1,u_2)\in E_G, (v_1,v_2)\in E_H \}.$$
Clearly, for graphs $A$,$B$ and $C$, we have that $(V_A\times V_B)\times V_C$ and $V_A\times (V_B\times V_C)$ are in bijection and thus, up to isomorphism, the cross product is associative. Hence, for each $m\geq 1$, we can define  $H^m=\underbrace{H\times \ldots \times H}_{m \text{ times}}$.  Last, we need the following graph parameter, defined with respect to the smallest odd-cycle in the graph.

\begin{definition}
    The {\em odd-girth} of a non-bipartite graph $H$ (denoted by $og(H)$) is the size of a smallest odd-cycle induced in $H$.
\end{definition}

For a bipartite graph we define the odd-girth to be infinite.

\subsection{Polymorphisms, pp/qfpp-definitions}

Even though the previous definitions apply only to graphs, {\it i.e.}, binary symmetric and irreflexive relations, we will need to introduce the following notions for relations $R$ of arbitrary arity.

\begin{definition}
Let $V$ be a finite set, $n,m\geq 1$ be integers, and let $ R\subseteq V^n$ be an $n$-ary relation on $V$. A partial function $f:dom(f)\rightarrow V$, with $dom(f)\subseteq V^m$,  is said to be a {\em partial polymorphism of $R$} if for every $n\times m$ matrix $A=(A_{i,j})\in V^{n\times m}$ such that for every $j\in [m]$,  the $j$-th column $A_{*,j}\in  R$ and for every $i\in [n]$, the $i$-th row $A_{i,*}\in dom(f)$,  the column $(f(A_{1,*}),\ldots,f(A_{n,*}))^{\top}\in R$.
In the case when $dom(f)=V^m$,  $f$ is a said to be a {\em total polymorphism} (or just a {\em polymorphism}) {\em of $R$}. We denote the sets of total and partial polymorphisms of $R$ by $\pol( R)$ and  $\pPol( R)$, respectively.
\end{definition}

Every (partial) function over a set $V$ is a (partial) polymorphism of both the empty relation (denoted by $\emptyset$) and the equality relation $\text{EQ}_V=\{(x,x) \mid x \in V\}$ over $V$ (or simply $\text{EQ}$ when the domain is clear from the context).

For a graph $H$, we sometimes use $\pPol(H)$ and $\pol(H)$ instead of $\pPol(E_H)$ and $\pol(E_H)$, where $E_H$ is viewed as a binary relation over the domain $V_H$. Note that $\pol(H)$ is exactly the set of $H$-colorings of $H^m$ for $m\geq 1$, and that $\pPol(H)$ is exactly the set of $H$-colorings of the induced subgraphs of $H^m$ for $m\in\N$.

\begin{definition}
Let  $R$ be a relation over a finite domain $V$. 
An $n$-ary relation $R'$
over $V$ is said
to have a {\em primitive positive-definition} (pp-definition) w.r.t. $R$ if there exists $m,m',n'\in\N$ such that
\begin{multline}
  R'(x_1, \ldots, x_{n}) \equiv \exists x_{n+1}, \ldots,
x_{n+n'} \colon \\ R(\mathbf{x}_1) \wedge \ldots \wedge
R({\mathbf{x}_m})\wedge \text{EQ}(\mathbf{y}_1) \wedge \ldots \wedge \text{EQ}(\mathbf{y}_{m'})\label{eq:def}
\end{multline}
where each $\mathbf{x}_i$ is an $\ar(R)$-ary tuple  and each $\mathbf{y}_i$ is an binary tuple of variables from $x_1,\ldots, x_{n}$, $x_{n+1}, \ldots, x_{n+ n'}$. Each term of the form $R(\mathbf{x}_i)$ or $\text{EQ}(\mathbf{y}_{j})$ is called an {\em atom} or a {\em constraint} of the pp-definition \eqref{eq:def}.
\end{definition}

In addition, if $n'=0$, then \eqref{eq:def} is called a {\em quantifier-free primitive positive-definition} (qfpp-definition) of $R'$. Let $\qfppclone{R}$ and $\cclone{ R}$  be the sets of qfpp-definable and  of pp-definable, respectively,  relations  over
$R$.

\begin{theorem}[\cite{poschel2013galois}]\label{thm:hgalois}\label{pp/qfpp Pol/pPol}
  Let $R$ and $R'$ be two relations over the same finite domain. Then 
  
  \begin{enumerate}
  \item $R'\in \qfppclone{ R}$ if and only if $\pPol( R)
  \subseteq \pPol(R')$ and 
  \item $R' \in \cclone{ R}$
  if and only if $\pol( R) \subseteq \pol(R')$.
  \end{enumerate}
\end{theorem}

\subsection{\Csp s and polymorphisms}

We now recall the link between the complexity of \Csp s and the algebraic tools described in the previous section (recall that the \textsc{$H$-Coloring} problem is the same problem as \Csp($E_H$)).

\begin{theorem}[\cite{jeavons1998}]\label{Poly reduction}
  Let $R$ and $R'$ be two relations over the same finite domain where $\pol( R)\subseteq \pol(R')$. Then {\em \Csp($R'$)} is polynomial-time many-one reducible to {\em \Csp($R$)}.
\end{theorem}

  Let $R$ be a relation over a finite domain $V$. Define:
  \[{\sf T}( R)=\inf \{c>1~\colon \text{ \Csp($R$) is solvable in time}~ O^*(c^n)\}\]
  where $n$ is the number of variables in a given \Csp($R$) instance, (with the notation $O^*(v_n)=O(v_n\times n^{O(1)})$ for all $(v_n)_{n\in\N}\in \mathbb{R}^{\N}$).

\begin{theorem}[\cite{jonsson2017}] \label{thm:ppol_complexity}
  Let $R$ and $R'$ be relations over a finite domain $V$. If
  $\pPol( R) \subseteq \pPol(R')$, then ${\sf T}(R') \leq
  {\sf T}( R)$.
\end{theorem}

These two theorems motivate our study of polymorphisms of graphs: since \Csp$(E_H)$ is the same problem as \textsc{$H$-Coloring}, key information about the fine-grained complexity of \textsc{$H$-Coloring} is contained in the set $\pPol(E_H)$.

\section{The inclusion structure of partial polymorphisms of graphs} \label{sec:classification}

In this section we study the inclusion structure of sets of the form $\pPol(H)$ when $H$ is a graph with $V_H=V$ for a fixed, finite set $V$. In other words, we are interested in describing the set $${\mathcal H} = \{\pPol(H) \mid H \text{ is a graph over }V\}$$ partially ordered by set inclusion. Here, one may observe that the requirement that $V_H = V_{H'} = V$ is not an actual restriction. Indeed, if $V_{H'} \subsetneq V$, then we can easily obtain a graph over $V$ simply by adding isolated vertices, with no impact on the set of partial polymorphisms.

\subsection{Trivial inclusion structure}

Our starting point is to establish $\pPol(H) \subseteq \pPol(H')$ when $H,H'$ are non-bipartite graphs, since it implies that (1) \textsc{$H$-Coloring} and $H'$-\textsc{Coloring} are both \NP-complete, and (2)  ${\sf T}(H') \leq {\sf T}(H)$, i.e., that $H'$-\textsc{Coloring} is not strictly harder than \textsc{$H$-Coloring}.

Inclusions of this kind e.g.\ raise the question whether there
 exist, for every fixed finite domain $V$, an \NP-hard $H$-\textsc{Coloring} problem which is (1) maximally easy, or (2) maximally hard\footnote{Here, ``maximally'' refers to the function {\sf T}.}.

As we will soon prove, the set ${\mathcal H}$ does {\em not} admit any non-trivial inclusions, in the sense that $\pPol(H)\subseteq \pPol(H')$ implies that either $H=H'$ or $E_{H'}=\emptyset$, for {\em all} $\pPol(H), \pPol(H') \in {\mathcal H}$.

\begin{theorem} \label{thm:no_inclusions}
  Let $H$ and $H'$ be two graphs with the same finite domain $V_H=V_{H'}=V$. Then $\pPol(H)\subseteq \pPol(H')$ if and only if $H=H'$ or $E_{H'}=\emptyset$.
\end{theorem}

\begin{proof}
  To prove sufficiency, assume that $H=H'$ or that $H'$ has no edges. Then $\pPol(H)=\pPol(H')$ or
  $\pPol(H)\subseteq \pPol(H')$ since in the latter case $\pPol(H')$ contains every partial function.

To prove necessity, assume that $\pPol(H)\subseteq \pPol(H')$. Then, by Theorem~\ref{pp/qfpp Pol/pPol}, $E_H$ qfpp-defines $E_{H'}$. However, the only possible atoms using $E_H$ and two variables $x$ and $y$ are: (1) $E_H(x,x)$ and $E_H(y,y)$, which cannot appear by irreflexivity, unless $E_{H'}=\emptyset$ and (2) $E_H(x,y)$ and $E_H(y,x)$, which are equivalent by symmetry. Also, if the qfpp-definition would contain an equality constraint $\text{EQ}(x,y)$, then $E_{H'}$ would not be irreflexive, unless $E_{H'}=\emptyset$. Hence, any qfpp-definition of $E_{H'}$ either (1) contains $E_{H}(x,x)$, $E_{H}(y,y)$ or $\text{EQ}(x,y)$, meaning that $E_{H'} = \emptyset$, or (2) only contains $E_{H}(x,y)$ or $E_H(y,x)$, meaning that $H = H'$.
\end{proof}

\subsection{Higher-arity inclusions}

 As proven in Theorem~\ref{thm:no_inclusions}, the expressivity of binary irreflexive symmetric relations is rather limited, in the sense that ${\mathcal H}$  does not admit {\em any} non-trivial inclusions.
 It is thus natural to  ask whether anything at all can be said concerning inclusions of the form $\pPol(H) \subseteq \pPol(R)$ when $R$ is an arbitrary relation. In particular, under which conditions does there exist an $n$-ary $R$ such that $\pPol(H) \subsetneq \pPol( R)$, given that \textsc{$H$-Coloring} and \Csp$( R)$ are both \NP-complete?
We give a remarkably sharp classification and, assuming that $\textsf{P} \neq \NP$, we prove that an $n$-ary relation $R$ with the stated properties exists if and only if $H$ contains an odd-cycle of length $\leq n$. We first require the following auxiliary lemma (recall the definition of odd-girth from Section~\ref{sec:prel}).

\begin{lemma}\label{lem:ppdefog}
Let $H$ be a non-bipartite graph and let $k:=og(H)$. Let $(x_1,\dots,x_k)\in (V_H)^k$. Then, $x_1x_2\dots x_kx_1$ forms an induced $k$-cycle in $H$ if and only if $(x_1,x_2), (x_2,x_3),\dots (x_{k-1},x_k)$ and $(x_k,x_1)$ are edges in $H$.
\end{lemma}

\begin{proof}

First, notice that if $x_1x_2\dots x_kx_1$ forms an induced $k$-cycle in $H$, then $(x_1,x_2), (x_2,x_3),\dots (x_{k-1},x_k)$ and $(x_k,x_1)$ are edges in $H$, which proves sufficiency.

To prove necessity, assume that $(x_1,x_2), (x_2,x_3),\dots (x_{k-1},x_k)$ and $(x_k,x_1)$ are edges in $H$.
Consider a smallest odd-cycle $C=x'_1\dots x'_px'_1$ (with $p$ odd and $p\le k$) that is a subgraph (not necessarily induced) of $H[\{x_1\dots x_k\}]$. Such an odd-cycle exists because $x_1\dots x_kx_1$ is an odd-cycle. We prove by contradiction that $C$ is induced in $H$. If $C$ is not induced in $H$, there exists an edge $(x'_i,x'_j)\in E_H$ with $i<j$ and $j-i \notin \{1,p-1\}$. Then either $j-i$ is odd and $x'_1\dots x'_{i-1}x'_ix'_jx'_{j+1}\dots x'_px'_1$ is an odd-cycle (of size $p-(j-i)+1<p$) smaller than $C$ that is a subgraph of $H[\{x_1\dots x_k\}]$, which contradicts the definition of $C$, or $j-i$ is even and $x'_ix'_{i+1}\dots x'_jx'_i$ is an odd-cycle (of size $(j-i)+1<p$) smaller than $C$ that is a subgraph of $H[\{x_1\dots x_k\}]$, which contradicts the definition of $C$.

By definition of $k:=og(H)$, the induced odd-cycle $C$ has at least $k$ vertices. Since $C$ is an induced subgraph of $H[\{x_1\dots x_k\}]$ with $k$ vertices that is an odd-cycle, $x_1\dots x_kx_1$ induces a $k$-cycle, which proves necessity.
\end{proof}

The following definition and lemma are particularly usefull when establishing our classification.

\begin{definition}\label{def:wall}
Let $n,m\geq 1$ be integers, $H$ be a graph, $R$ be a relation of arity $n$ over $V_H$, and let $M=(M_{i,j})$ be a  $n\times m$ matrix of elements of $V_H$. We say that $M$ is an {\em $R$-wall for $H$} if:
\begin{enumerate}
    \item $\forall j\in [m],\ (M_{1,j},...,M_{n,j})^{\top}\in R$, and
    \item $\forall (i,i')\in [n]^2,\ \exists j\in [m],\ (M_{i,j},M_{i',j})^{\top}\notin E_H$.
\end{enumerate}
\end{definition}

In the following lemma, we say, for a relation $R$, that \Csp($R$) is {\em trivial} if every instance of \Csp($R$) is satisfiable. Clearly, if \Csp($R$) is trivial, it is not \NP-complete, even if \textsf{P=NP}.

\begin{lemma}\label{prop:Triviality}
Let $H$ be a graph and let $R$ be  an $n$-ary  relation over $V_H$. Suppose that $\pPol(H)\subseteq \pPol( R)$ and that there exists an $R$-wall $M$ for $H$. Then, \Csp($R$) is trivial.
\end{lemma}

\begin{proof}
Using property 2) of Definition~\ref{def:wall}, it is easy to check that any partial function $f$ whose domain is the set of rows of $M$ is in $\pPol(H)$. 
In particular, $f$ can be chosen to be of constant value $a\in V_H$.
Then, from $\pPol(H)\subseteq \pPol( R)$ it follows that $f\in \pPol(R)$. Combining this with property 1) of Definition~\ref{def:wall}, we conclude that $(a,\ldots,a)^{\top}\in  R$. 
Since the valuation sending all variables to $a$  satisfies any instance of \Csp($R$), the proof is now complete.
\end{proof}

We now propose a construction of an $R$-wall for a graph $H$ with $n:=\ar(R) < og(H)$, and  such that $\pPol(H)\subsetneq\pPol(R)$.

\begin{lemma}\label{lem:qfpp}
Let  $H$ be a graph with $og(H)> n$, and let $ R \neq \emptyset$ be an $n$-ary relation such that $\pPol(H)\subseteq \pPol( R)$. If $\forall (x_1,\ldots,x_n)\in (V_H)^n,  R(x_1,\ldots,x_n) \implies E_H(x_1,x_2)$, then $R$ qfpp-defines $E_H$.
\end{lemma}

\begin{proof}
Suppose that $ R(x_1,\ldots,x_n) \implies E_H(x_1,x_2)$, and let $(a_1,\ldots,a_n)\in R \neq \emptyset$. Since $H$ has no odd-cycle of size $\leq n$, $\{a_1,\ldots a_n\}$ induces a bipartite graph in $H$: there exists a partition $A\uplus B$ of $\{a_1,\ldots,a_n\}$ such that $E_{H[\{a_1,\ldots,a_n\}]}\subseteq (A\times B)\uplus (B\times A)$.
For $(x,y)\in E_H$, define $f_{x,y}\colon \{a_1,\ldots,a_n\}  \rightarrow  V_H$ by
$f_{x,y}(a_i)=x$, if $a_i\in A$, and $f_{x,y}(a_i)=y$, if $a_i\in B$.
Since $(x,y)\in E_H$, we have that $f_{x,y}\in\pPol(H)$, and since $\pPol(H)\subseteq \pPol( R)$, we also have that $f_{x,y}\in \pPol( R)$. As $(a_1,\ldots,a_n)\in R$, it follows $(f_{x,y}(a_i))_{1\leq i\leq n}\in R$.

This proves that $E_H(x,y)\implies  R(\textbf{x}_{A,B}(x,y))$, where $\textbf{x}_{A,B}(x,y)[i]:=f_{x,y}(a_i)$ equals $x$ if $a_i\in A$ and $\textbf{x}_{A,B}(x,y)[i]=y$ if $a_i\in B$.

Reversely, since $(a_1,\dots,a_n)\in R$ and $ R(x_1,\ldots,x_n) \implies E_H(x_1,x_2)$, we have $(a_1,a_2)\in E_H$. Recall that $E_H\subseteq (A\times B)\uplus (B\times A)$: it follows by definition of $\textbf{x}_{A,B}(x,y)$ that for all vertices $x$ and $y$ of $H$, that $\{\textbf{x}_{A,B}(x,y)[1],\textbf{x}_{A,B}(x,y)[2]\}=\{x,y\}$. From the fact that $E_H$ is symmetric and the hypothesis that $ R(x_1,\ldots,x_n) \implies E_H(x_1,x_2)$, we deduce that $R(\textbf{x}_{A,B}(x,y))\implies E_H(x,y)$. Hence, $E_H(x,y)\equiv  R(\textbf{x}_{A,B}(x,y))$, and $R$ qfpp-defines $E_H$.
\end{proof}

\begin{lemma}\label{Existence lemma}
Let $n\geq 1$, $H$ be a graph with  $og(H)> n$, and let $ R \neq \emptyset$ be an $n$-ary relation such that  $\pPol(H)\subsetneq \pPol( R)$. Then, for all $(i,i')\in [n]^2$ with $i<i'$, there is $(x_1^{(i,i')},\ldots,x_n^{(i,i')})^{\top}\in R$ with $(x_i^{(i,i')},x_{i'}^{(i,i')})^{\top}\notin E_H$.

\end{lemma}

\begin{proof}
We show only the existence for $i=1$ and $i'=2$; the other cases can be proven similarly. 
For the sake of a contradiction, suppose  that $\forall (x_1,\ldots,x_n)\in (V_H)^n, (x_1,\ldots,x_n)\in R \implies (x_1,x_2)\in E_H$. By Lemma~\ref{lem:qfpp} we have $E_H\in \langle  R \rangle_{\cancel{\exists}}$, and  by Theorem~\ref{pp/qfpp Pol/pPol}, $\pPol( R)\subseteq \pPol(E_H)$. This contradicts our hypothesis that $\pPol(H)\subsetneq \pPol( R)$.
\end{proof}

This leads to the following corollary whose proof provides a simple construction of an  $R$-wall for graph $H$ in the conditions of Lemma~\ref{Existence lemma}.

\begin{corollary}\label{corollary:wall}
Let $n\geq 1$,  $H$ be a graph with $og(H)> n$, and let $ R \neq \emptyset$ be an $n$-ary relation such that  $\pPol(H)\subsetneq \pPol( R)$. Then there is an $R$-wall for $H$.
\end{corollary}

\begin{proof}
Using the notation of Lemma~\ref{Existence lemma}, we can take the $n\times \frac{n(n-1)}{2}$ matrix $M$, whose $\frac{n(n-1)}{2}$ columns are the $(x_1^{(i,i')},\ldots,x_n^{(i,i')})^{\top}$, for each $(i,i')\in [n]^2$ with $i<i'$.
\end{proof}

We are now ready to prove the main result of this section.

\begin{theorem}
Let $H$ be a graph, and let $k:=og(H)$.  There exists an $n$-ary relation $R\neq\emptyset$ with $\pPol(H)\subsetneq \pPol( R)$ such that {\em \Csp($R$)} is \NP-complete if and only if $k\leq n$. Moreover, if $k>n$, any $n$-ary relation $R\neq\emptyset$ with $\pPol(H)\subsetneq \pPol( R)$ is such that {\em \Csp($R$)} is trivial.
\end{theorem}

\begin{proof}
We sketch the most important ideas.

 Suppose first that $k>n$. In this case, $H$ does not have an odd-cycle of length $\leq n$. Again for the sake of a contradiction, suppose that such a relation $R$ exists. 
Note that $ R\neq\emptyset$ since \Csp($R$) is \NP-complete.
Using Corollary~\ref{corollary:wall}, there exists an $R$-wall for $H$. Then, by Lemma~\ref{prop:Triviality}, \Csp($R$) is trivial. This contradicts the fact that \Csp($R$) is \NP-complete, and thus such a relation $R$ does not exist.

 Suppose now that $k\leq n$. Define  $ R(x_1,\ldots,x_n) \equiv E_H(x_1,x_2) \wedge E_H(x_2,x_3) \wedge \ldots \wedge E_H(x_{k-1},x_k) \wedge E_H(x_k,x_1)$.
Since $k=og(H)$, it follows from Lemma~\ref{lem:ppdefog} that $ R = \{(x_1, \ldots, x_n) \mid (x_1, \ldots, x_k) \text{ forms a } k\text{-cycle} \}$ (the variables $x_{k+1},\ldots,x_n$ are inessential).
  
  We then proceed as follows. Since $E_H$ qfpp-defines $R$, we have  that $\pPol(H)\subseteq \pPol( R)$, by Theorem~\ref{pp/qfpp Pol/pPol}.
Also, the inclusion is strict since, for any edge $(x,y)$ of $H$, the function  $f: \{x,y\}  \to  V_H$ that maps both $x$ and $y$ to any $a\in V_H$, belongs to  $\pPol( R)\setminus \pPol(H)$. Indeed, $f\in\pPol( R)$ because $\{x,y\}^n \cap  R =\emptyset$, since it is impossible to form an odd-cycle with only $x$ and $y$.
      
 To prove that \Csp($R$) is \NP-complete, consider $C_k(H)$, the subgraph of $H$, with $V_{C_k(H)}=V_H$, and where each edge of $H$ that does not belong to a cycle of length $k$ has been removed. Note that as $H$ contains a $k$-cycle, $C_k(H)$ also contains a $k$-cycle, and is therefore non-bipartite. Hence, \Csp($E_{C_k(H)}$), which is the same problem as the $C_k(H)$-\textsc{Coloring} problem, is \NP-hard (by Theorem~\ref{thm:Complexity H-COLORING}).
 
 It is easy to see that $R$ pp-defines $E_{C_k(H)}$, as
 
 $$E_{C_k(H)}(x_1,x_2) \equiv \exists x_3,\ldots,x_n,  R(x_1,x_2,x_3\ldots,x_n).$$
 
From Theorem~\ref{pp/qfpp Pol/pPol} and \ref{Poly reduction}, \Csp($E_{C_k(H)}$)=$C_k(H)$-\textsc{Coloring} has a polynomial-time reduction to \Csp($R$), and \Csp($R$) is thus \NP-hard. Clearly, it is also included in \NP. 
\end{proof}

\section{Projective and core graphs}\label{section:Projective graphs}

In this section we study the inclusion structure of sets of total polymorphisms. We are particularly interested in graphs $H$ with small sets of polymorphisms since, intuitively, they correspond to the hardest \textsc{$H$-Coloring} problems. This motivates the following definitions.

An $m$-ary function $f$ is said to be {\em essentially at most unary} if it is of the form $f=f'\circ \pi_i^m$ for some $i\in [m]$ and some unary function $f'$.
Larose~\cite{larose2002families} says that a graph $H$ is {\em projective} if every {\it idempotent polymorphism} (i.e., $f(x, \ldots, x) = x$ for every $x \in V_H$) is a projection. 
Okrasa and Rz\c ażewski~\cite{okrasa2020finegrained} showed that the  polymorphisms of a core graph $H$ are all essentially at most unary if and only if $H$ is projective. %
Since it is sufficient to study cores in the context of \textsc{$H$-Coloring}, determining whether $H$ is projective is particularly interesting.

In this section we use the algebraic approach for proving that a given graph is a projective core, that is,  both projective and a core. As we will see, this enables simpler proofs than those of \cite{larose2002families}, and suggests the possibility of completely characterizing projective cores.

Using the following theorem, our proofs of projectivity can be seen as reductions from cliques.

\begin{theorem}[\cite{bodirsky2015graph,luczak2004note}]\label{thm:CliqueProjective}
For $k\geq 3$, $K_k$ is projective.
\end{theorem}

Let $\mathfrak{S}_k$ be the set of permutations over $[k]$. It then follows that 
$
\pol(K_k) = \{ \sigma\circ\pi_i^m \mid \sigma\in\mathfrak{S}_k, m\geq 1, i\in [m] \}
$. The inclusion $\supseteq$ is indeed clear. To justify $\subseteq$, note that if $f\in \pol(K_k)$ with $\ar(f)=m$, the function $\sigma:x\mapsto f(x,\dots,x)$ is a unary polymorphism of $K_k$, and is therefore bijective: $\sigma$ is an automorphism of $K_k$ ie. $\sigma\in\mathfrak{S}_k$. Then, since $\sigma^{-1}\circ f$ is a polymorphism of $K_k$ (by composition of polymorphisms of $K_k$) that is idempotent, it is a projection $\pi_i^m$ with $i\in [m]$ by Theorem \ref{thm:CliqueProjective}, and then $f=\sigma\circ \pi_i^m$.

Corollary~\ref{cor:projectivity} below implies that the graphs we will consider in this subsection are projective cores.

\begin{corollary}\label{cor:projectivity}
Let $H$ be a graph on $[k]$ with $k\geq 3$.  Then $E_H$ pp-defines the relation $NEQ_k =\{ (x,x')\in V_H \mid x\neq x'\}$ if and only if $H$ is a projective core.
\end{corollary}

\begin{proof} First observe  that $\text{NEQ}_k=E_{K_k}$. From Theorem \ref{pp/qfpp Pol/pPol} and using the definitions of cores and of projective graphs, we thus have that the following assertions are equivalent:
\begin{enumerate}
    \item $\text{NEQ}_k\in \cclone{E_H}$;
    \item $\pol(H)\subseteq\pol(K_k)$;%
    \item all polymorphisms of $H$ are essentially at most unary, and all unary polymorphisms of $H$ are bijective;
    \item $H$ is a projective core.\qedhere
\end{enumerate}
\end{proof}

By following the  steps in the proof of Corollary~\ref{cor:projectivity}  we can obtain the  following result.

\begin{corollary}
Let $G$ and $H$ be two graphs on the same set of vertices, with $G$ projective (respectively, a core), and such that $E_H$ that pp-defines $E_G$. Then $H$ is also projective (respectively, a core).
\end{corollary}

Pp-definitions thus explains the property of being projective (respectively, a core). We hope that this viewpoint helps to discover new classes of projective graphs.
 For example, Corollary~\ref{cor:projectivity} enables a much simpler proof of the following theorem by Larose~\cite{larose2002families}.

\begin{theorem}[\cite{larose2002families},\cite{larose2001strongly}] \label{thm:cycle_core}
Let $k\geq 3$ be an odd integer. The $k$-cycle $C_k$ is a projective core.
\end{theorem}

\begin{proof}
We claim that
 \begin{eqnarray*}
 \text{NEQ}_k(x,x')&\equiv \exists x_2,\ldots,x_{k-2} \colon E_{C_k}(x,x_2) \wedge E_{C_k}(x_2,x_3)\\ &\wedge \ldots \wedge E_{C_k}(x_{k-3},x_{k-2}) \wedge E_{C_k}(x_{k-2},x').
 \end{eqnarray*}
To see this, note that for any two vertices $x$ and $x'$ in $C_k$,  $x\neq x'$ if and only if there exists an odd-path from $x$ to $x'$ of size $<k$  (since $k$ is odd). In other words, $x\neq x'$ if and only if there exists a $(k-2)$-path from $x$ to $x'$ (by going through the same edge as many times as necessary, $k-2$ being odd).
By Corollary~\ref{cor:projectivity}, it then follows that $C_k$ is a projective core.
\end{proof}

There are also other examples of cores that are projective, other than $k$-cliques for $k\geq 3$ and $k$-cycles. For instance,
Okrasa and Rz\c ażewski~\cite{okrasa2020finegrained} proved that the so-called {\em Grötzsch graph} (see Figure~\ref{The Grötzsch graph}) is a projective core.
\begin{center}
\begin{figure}
\centering

\begin{tikzpicture}[scale=.5]
\tikzstyle{every node}=[draw,shape=circle];

\begin{scope}
    \node (v1) at (0+90:2) {};
    
    \node (v2) at (-72+90:2) {}
        edge (v1);
        
    \node (v3) at (-2*72+90:2) {}
        edge (v2);

    \node (v4) at (-3*72+90:2) {}
        edge (v3);

    \node (v5) at (-4*72+90:2) {}
        edge (v1)
        edge (v4);
    
    \node (v1') at (0+90:1) {}
        edge (v2)
        edge (v5);
    
    \node (v2') at (-72+90:1) {}
        edge (v1)
        edge (v3);
        
    \node (v3') at (-2*72+90:1) {}
        edge (v2)
        edge (v4);

    \node (v4') at (-3*72+90:1) {}
        edge (v3)
        edge (v5);

    \node (v5') at (-4*72+90:1) {}
        edge (v1)
        edge (v4);
    
    \node (u) at (0:0) {}
        edge (v1')
        edge (v2')
        edge (v3')
        edge (v4')
        edge (v5');
\end{scope}

\begin{scope}[xshift=6cm]

    \node (v1) at (0+90:2) {};
    
    \node (v2) at (-72+90:2) {}
        edge (v1);
        
    \node (v3) at (-2*72+90:2) {}
        edge (v2);

    \node (v4) at (-3*72+90:2) {}
        edge (v3);

    \node (v5) at (-4*72+90:2) {}
        edge (v1)
        edge (v4);
    
    \node (v1') at (0+90:1) {}
        edge (v1);
    
    \node (v2') at (-72+90:1) {}
        edge (v2);
        
    \node (v3') at (-2*72+90:1) {}
        edge (v3)
        edge (v1');

    \node (v4') at (-3*72+90:1) {}
        edge (v4)
        edge (v1')
        edge (v2');

    \node (v5') at (-4*72+90:1) {}
        edge (v5)
        edge (v2')
        edge (v3');

\end{scope}

\end{tikzpicture}
\caption{The Grötzsch graph (left) and the Petersen graph (right)}
\label{The Grötzsch graph}
\end{figure}
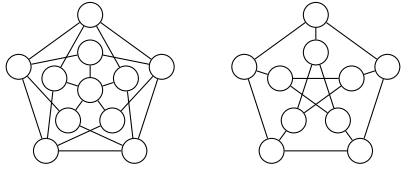
\end{center}
\vspace{-.5cm}
\begin{theorem}%
The Grötzsch and Petersen graph is a projective core\footnote{We acknowledge Mario Valencia-Pabon for pointing out that the proof for the Gr\"otzsch graph also applies to the Petersen graph.}
\end{theorem}

\begin{proof}
We provide an alternative proof using our algebraic framework.
Let $E_G$ be the set of edges of the Grötzsch graph. Note that the Grötzsch graph has 11 vertices.
We can see that $E_G$ pp-defines $\text{NEQ}_{11}$: $$\text{NEQ}_{11}(x,x')\equiv \exists x_2,x_3 \colon  E_G(x,x_2) \wedge E_G(x_2,x_3) \wedge E_G(x_3,x').$$

From Corollary~\ref{cor:projectivity} it follows that the Grötzsch graph is a projective core. The proof for the Petersen graph is analogous. 
\end{proof}

Complements $\overline{C_k}$ of odd-cycles of length $k\geq 5$ are also projective cores. Since $\overline{C_5}=C_5$ has already been studied, we take a look at $\overline{C_{2p+1}}$, for $p\geq 3$. The following result is an immediate corollary of Larose~\cite{larose2002families}, but we give an algebraic proof using Corollary \ref{cor:projectivity}.

\begin{theorem}
\label{thm:ComplementCycleProjectiveCore} $\overline{C_{2p+1}}$ is a projective core for $p\geq 3$.
\end{theorem}

\begin{proof}
It is not difficult to see that
$\text{NEQ}_{2p+1}(x_1,x_4)\equiv \exists x_2,x_3,w_1,\ldots,w_{p-2}: R_1\wedge R_2 \wedge R_3,$ where
\begin{enumerate}
    \item $R_1 = \bigwedge\limits_{i\in [3]} E_{\overline{C_{2p+1}}}(x_i,x_{i+1})$, 
    \item $R_2=\bigwedge\limits_{i\in [4], j\in [p-2]} E_{\overline{C_{2p+1}}}(x_i,w_j)$, and 
    \item $R_3=\bigwedge\limits_{(j,j')\in [p-2]^2, j< j'} E_{\overline{C_{2p+1}}}(w_j,w_{j'})$. 
\end{enumerate}

The result then follows from Corollary~\ref{cor:projectivity}.
\end{proof}

Moreover, we can prove by Corollary \ref{cor:projectivity} that adding universal vertices to $C_5$ results in a projective core.

\begin{theorem}\label{thm:C5+pProjective}

Let $p\ge 0$. The graph $C_5+p$, obtained from $C_5$ by adding $p$ universal vertices\footnote{Formally $C_5+p = (V_{C_5}\uplus V_{K_p},E_{C_5}\uplus E_{K_p}\uplus (V_{C_5}\times V_{K_p}) \uplus (V_{K_p}\times V_{C_5}))$.} is a projective core.

\end{theorem}

\begin{proof}

We can see that $E_{C_5+p}$ pp-defines NEQ$_{p+5}$ though the pp-definition:

$\text{NEQ}_{p+5}(x_1,x_4)\equiv \exists x_2,x_3,w_1,\ldots,w_p: R_1\wedge R_2 \wedge R_3,$ where

\begin{enumerate}
    \item $R_1 = \bigwedge\limits_{i\in [3]} E_{C_5+p}(x_i,x_{i+1})$, 
    \item $R_2=\bigwedge\limits_{i\in [4], j\in [p]} E_{C_5+p}(x_i,w_j)$, and 
    \item $R_3=\bigwedge\limits_{(j,j')\in [p]^2, j< j'} E_{C_5+p}(w_j,w_{j'})$. 
\end{enumerate}

which proves that $C_5+p$ is a projective core by Corollary \ref{cor:projectivity}.

To see that the pp-definition is correct, notice that if $x_1=x_4$, the pp-definition can not be satisfied, since it would imply the existence of a $K_{p+3}$ (induced by $x_1,x_2,x_3,w_1,\dots,w_p$) in $C_5+p$, which is absurd. Note also that if $x_1\neq x_4$ and $x_1$ and $x_4$ are adjacent, then there exists $y_1,\dots,y_p$ in $C_5+1$ such that $\{x_1,x_4,y_1,\dots,y_p\}$ induces a $K_{p+2}$: taking $x_2:=x_4$, $x_3:=x_1$ and $w_j:=y_j$ (for all $j\in [p]$) satisfies the pp-definition. Also, if $x_1\neq x_4$ and $x_1$ and $x_4$ are not adjacent, we can assume by symetry that $x_1=1$ and $x_4=4$ (where the vertices of the $C_5$ induced in $C_5+p$ are named $0,1,2,3,4$ in order). Then, taking $x_2:=2$, $x_3:=3$ and $w_1,\dots,w_p$ the $p$ universal vertex of $C_5+p$ satisfies the pp-definition. The pp-definition is thus correct.
\end{proof}

\section{Verifying the conjecture on small graphs}

Okrasa and Rz\c ażewski~\cite{okrasa2020finegrained} observed that a graph $H$ that can be expressed as a disjoint union of two non-empty graphs $H_1$ and $H_2$ is not projective, since it admits the binary polymorphism $f$ defined by $f|_{V_{H_1}\times V_H}=(\pi_1^2)|_{V_{H_1}\times V_H}$ and $f|_{V_{H_2}\times V_H}=(\pi_2^2)|_{V_{H_2}\times V_H}$. 
The same holds for the  {\em cross-product} of non-trivial graphs $H=H_1\times H_2$ (in which case the graph is said to be decomposable), with the binary polymorphism $f((x_1,x_2),(y_1,y_2))\mapsto (x_1,y_2)$.
Okrasa and Rz\c ażewski also noticed the existence of disconnected cores, such as $G+ K_3$ (indecomposable cores are much more difficult to study), where $G$ is the Grötzsch graph from Figure \ref{The Grötzsch graph}. %
These observations resulted in the following conjecture.

\Okrasa*

The goal of this section is to apply the tools constructed in Section \ref{section:Projective graphs} to the verification of the conjecture below by Okrasa and  Rz\c ażewski~\cite{okrasa2020finegrained}. First, in Section \ref{sec:Conjecture6Vertices}, we completely classify the cores on at most $6$ vertices and verify that the Okrasa and Rz\c ażewski Conjecture is true on each of these graphs. Then, in Section \ref{sec:Conjecture7Vertices}, after giving an exhaustive list of the cores on 7 vertices, we prove the projectivity of all these graphs.

\begin{center}
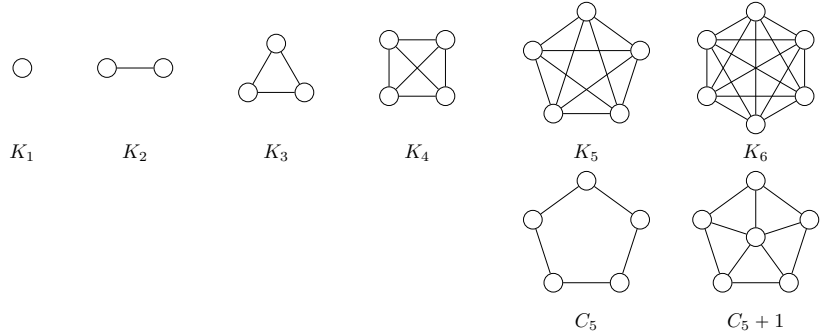
\begin{figure}
\centering
\scalebox{0.75}{
\begin{tikzpicture}

\tikzstyle{vertex}=[draw,shape=circle];

\begin{scope}

    \node[vertex] (v1) at (0,0) {};
    \node () at (0,-1.5) {$K_1$};

\end{scope}

\begin{scope}[xshift=2cm]

    \node[vertex] (v1) at (-0.5,0) {};
    \node[vertex] (v2) at (+0.5,0) {}
        edge (v1);
    \node () at (0,-1.5) {$K_2$};

\end{scope}

\begin{scope}[xshift=4.5cm]

    \node[vertex] (v1) at (-0.5,-1.732/4) {};
    \node[vertex] (v2) at (0.5,-1.732/4) {}
        edge (v1);
    \node[vertex] (v3) at (0,1.732/4) {}
        edge (v1)
        edge (v2);
    \node () at (0,-1.5) {$K_3$};

\end{scope}

\begin{scope}[xshift=7cm]

    \node[vertex] (v1) at (-0.5,-0.5) {};
    \node[vertex] (v2) at (0.5,-0.5) {}
        edge (v1);
    \node[vertex] (v3) at (-0.5,0.5) {}
        edge (v1)
        edge (v2);
    \node[vertex] (v4) at (0.5,0.5) {}
        edge (v1)
        edge (v2)
        edge (v3);
    \node () at (0,-1.5) {$K_4$};

\end{scope}

\begin{scope}[xshift=10cm]

    \node[vertex] (v1) at (90:1) {};
    \node[vertex] (v2) at (-72+90:1) {}
        edge (v1);
    \node[vertex] (v3) at (-2*72+90:1) {}
        edge (v1)
        edge (v2);
    \node[vertex] (v4) at (-3*72+90:1) {}
        edge (v1)
        edge (v2)
        edge (v3);
    \node[vertex] (v5) at (-4*72+90:1) {}
        edge (v1)
        edge (v2)
        edge (v3)
        edge (v4);
    \node () at (0,-1.5) {$K_5$};

\end{scope}

\begin{scope}[xshift=13cm]

    \node[vertex] (v1) at (90:1) {};
    \node[vertex] (v2) at (-60+90:1) {}
        edge (v1);
    \node[vertex] (v3) at (-2*60+90:1) {}
        edge (v1)
        edge (v2);
    \node[vertex] (v4) at (-3*60+90:1) {}
        edge (v1)
        edge (v2)
        edge (v3);
    \node[vertex] (v5) at (-4*60+90:1) {}
        edge (v1)
        edge (v2)
        edge (v3)
        edge (v4);
    \node[vertex] (v6) at (-5*60+90:1) {}
        edge (v1)
        edge (v2)
        edge (v3)
        edge (v4)
        edge (v5);
    \node () at (0,-1.5) {$K_6$};

\end{scope}

\begin{scope}[xshift=10cm,yshift=-3cm]

    \node[vertex] (v1) at (90:1) {};
    \node[vertex] (v2) at (-72+90:1) {}
        edge (v1);
    \node[vertex] (v3) at (-2*72+90:1) {}
        edge (v2);
    \node[vertex] (v4) at (-3*72+90:1) {}
        edge (v3);
    \node[vertex] (v5) at (-4*72+90:1) {}
        edge (v1)
        edge (v4);
    \node () at (0,-1.5) {$C_5$};

\end{scope}

\begin{scope}[xshift=13cm,yshift=-3cm]

    \node[vertex] (v1) at (90:1) {};
    \node[vertex] (v2) at (-72+90:1) {}
        edge (v1);
    \node[vertex] (v3) at (-2*72+90:1) {}
        edge (v2);
    \node[vertex] (v4) at (-3*72+90:1) {}
        edge (v3);
    \node[vertex] (v5) at (-4*72+90:1) {}
        edge (v1)
        edge (v4);
    \node[vertex] (u) at (0,0) {}
        edge (v1)
        edge (v2)
        edge (v3)
        edge (v4)
        edge (v5);
    
    \node () at (0,-1.5) {$C_5+1$};

\end{scope}

\end{tikzpicture}
}
\caption{Every core graph with at most 6 vertices}
\label{fig:CoresAtMost6Vertices}
\end{figure}
\end{center}

This section aims at verifying the Okrasa and Rz\c ażewski Conjecture on graphs with at most $7$ vertices, and culminates with the proof of the following theorem:

\begin{restatable}{theorem}{ConjectureTrueSevenVertices}\label{thm:ConjectureTrueSevenVertices}

The Okrasa and Rz\c ażewski Conjecture is true on graphs with at most $7$ vertices

\end{restatable}

\subsection{Core graphs with at most 6 vertices}\label{sec:Conjecture6Vertices}

In order to verify the conjecture on small graphs, we enumerate all the (indecomposable) small cores and check their projectivity. Recall from Theorem \ref{thm:CliqueProjective} that cliques are indecomposable, projective and core, and thus the Okrasa and Rz\c ażewski Conjecture  is true on cliques. We can therefore restrict to the non-clique core graphs. This motivates the definition of {\em proper cores}.

\begin{definition}

A {\em proper core} is a core graph that is not a clique.

\end{definition}

Moreover, a proper core on $n\ge 0$ vertices is called a {\em proper $n$-core}.

Recall that a graph $G$ is said to be {\em perfect} if for all induced subgraph $G'$ of $G$, the size of the largest clique of $G'$ equals the chromatic number of $G'$.

\begin{remark}\label{rem:CoreNotPerfect}

If a graph $G$ is a proper core then it is not a perfect graph.

\end{remark}

\begin{proof}

Assume by contradiction that $G$ is a perfect graph, and let $k$ be the chromatic number of $G$. Since $G$ is a perfect graph, $G$ has an induced $K_k$, and by definition of $k$, $G$ is $k$-colorable. Since $K_k$ is a core, it follows that $core(G)=K_k$, and since $G$ is a core, $G=core(G)=K_k$. We have that $G$ is a clique, which contradicts the hypothesis that $G$ is a proper core.

\end{proof}

Using the famous theorem of perfect graphs, we can drastically reduce the search space when trying to enumerate all proper $6$-cores.

\begin{theorem}[Theorem of perfect graphs]\cite{chudnovsky2003progress}\label{thm:TheoremOfPerfectGraphs}

A graph $G$ is perfect if and only if $G$ does not contain any induced $C_k$ or $\overline{C_k}$ for some $k$ odd and $k\ge 5$.

\end{theorem}

We can immediately deduce the following corollary from Remark~\ref{rem:CoreNotPerfect} and Theorem~\ref{thm:TheoremOfPerfectGraphs}.

\begin{corollary}\label{cor:ProperCoreInducedCycle}

Let $G$ be a proper core on $n\ge 1$ vertices, then $G$ contains an induced $C_k$ or $\overline{C_k}$ for some $k$ odd and $5\le k\le n$.

\end{corollary}

Via Corollary \ref{cor:ProperCoreInducedCycle} we can easily classify the cores on $\le 5$ vertices (remarking that $\overline{C_5}=C_5$).

\begin{remark}\label{rem:CoresAtMost5Vertices}

The cliques $K_1,\dots,K_5$ are cores. The other cores on $\le 5$ vertices are proper cores, so must contain an induced $C_5$. It follows from Corollary \ref{cor:ProperCoreInducedCycle} that the only proper $n$-core with $n\le 5$ is $C_5$.

\end{remark}

We have completely classified the cores on $\le 5$ vertices. We now extend this classification to the cores on $6$ vertices.

\begin{theorem}\label{thm:Cores6Vertices}

The only core graphs on $6$ vertices are $K_6$ and the graph $C_5+1$ presented in Figure \ref{fig:CoresAtMost6Vertices}. These two graphs are projective.

\end{theorem}

\begin{proof}

First, notice that $K_6$ and $C_5+1$ are projective cores by Theorem \ref{thm:CliqueProjective} and \ref{thm:C5+pProjective}.

We now prove that $K_6$ and $C_5+1$ are the only cores on 6 vertices. Assume by contradiction that there exists a proper core $G$ on $6$ vertices different from $C_5+1$. Then by Corollary \ref{cor:ProperCoreInducedCycle}, five of the six vertices of $G$ must induce a $5$-cycle, call them $a,b,c,d$ and $e$, and call $u$ the sixth vertex.

Since $G$ is not isomorphic to $C_5+1$, $u$ must not be a neighbor to (at least) one the vertex in $\{a,b,c,d,e\}$. Assume by symmetry that $u$ and $a$ are not neighbors. Notice that $G$ is $3$-colorable by coloring $a$ and $u$ with the color $1$; $b$ and $d$ with the color $2$; and $c$ and $e$ with the color $3$. We deduce that $G$ has no triangle, otherwise we would have $core(G)=K_3$, contradicting the fact that $G$ is a core. It follows that $u$ has at most $2$ non-adjacent neighbors, ie. the set of neighbors of $u$ is contained in a set of the form $\{\alpha,\beta\}$ where $\alpha$ and $\beta$ belong to $\{a,b,c,d,e\}$ and are non-adjacent. The vertices $\alpha$ and $\beta$ have a common neighbor $\gamma\in\{a,b,c,d,e\}$. The function that maps $u$ to $\gamma$ and that leaves the rest of the graph unchanged is a homomorphism from $G$ to the $C_5$ induced by $\{a,b,c,d,e\}$. This proves that $core(G)=C_5$, contradicting that $G$ is a core.

We have proven by contradiction that the only cores on $6$ vertices are $K_6$ and $C_5+1$, and that they are projective.
\end{proof}

The completeness of the classification of cores on at most $6$ vertices presented in Figure \ref{fig:CoresAtMost6Vertices} follows from Remark \ref{rem:CoresAtMost5Vertices} and Theorem \ref{thm:Cores6Vertices}. All of these graphs are projective, by Theorems \ref{thm:CliqueProjective} and \ref{thm:cycle_core} and \ref{thm:Cores6Vertices} and thus are not counter-example to the Okrasa and Rz\c ażewski Conjecture. Hence, we have proven that the Okrasa and Rz\c ażewski Conjecture is true on graphs with at most $6$ vertices, and we now continue with graphs on 7 vertices.

\subsection{Cores on 7 vertices}\label{sec:Conjecture7Vertices}

In order to put the Okrasa and Rz\c ażewski Conjecture to the test, we continue to enumerate the small cores. We provide the exhaustive list of cores on $7$ vertices in Figure \ref{fig:Trivial7cores} and \ref{fig:Sporadic7Cores}. The proof of the fact that this is indeed the exhaustive list of cores on 7 vertices is left to Appendix \ref{app:7cores}.

\begin{restatable}{theorem}{ClassificationSevenCores}\label{thm:Classification7Cores}

Up to isomorphism, there are exactly 10 cores graphs on $7$ vertices. They are the graphs $K_7,C_7,\overline{C_7},C_5+2$ presented in Figure \ref{fig:Trivial7cores}, and the graphs $G_1,\dots,G_6$ presented in Figure \ref{fig:Sporadic7Cores}.

\end{restatable}

\begin{proof}
See Appendix \ref{app:7cores}.
\end{proof}

\begin{center}
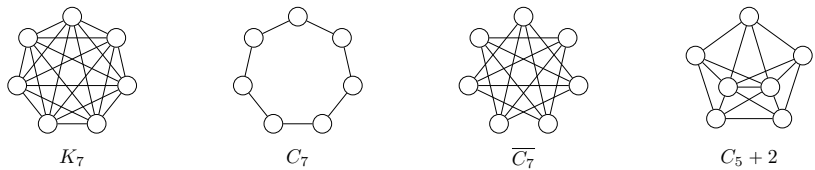
\begin{figure}
\centering
\scalebox{0.75}{
\begin{tikzpicture}

\tikzstyle{vertex}=[draw,shape=circle];

\begin{scope}

    \node[vertex] (a) at (90:1) {};
    \node[vertex] (e) at (-51.43+90:1) {}
        edge (a);
    \node[vertex] (d) at (-2*51.43+90:1) {}
        edge (a)
        edge (e);
    \node[vertex] (c) at (-3*51.43+90:1) {}
        edge (a)
        edge (e)
        edge (d);
    \node[vertex] (b) at (-4*51.43+90:1) {}
        edge (a)
        edge (e)
        edge (d)
        edge (c);
    \node[vertex] (v6) at (-5*51.43+90:1) {}
        edge (a)
        edge (e)
        edge (d)
        edge (c)
        edge (b);
    \node[vertex] (v7) at (-6*51.43+90:1) {}
        edge (a)
        edge (e)
        edge (d)
        edge (c)
        edge (b)
        edge (v6);
    \node () at (0,-1.5) {$K_7$};

\end{scope}

\begin{scope}[xshift=4cm]

    \node[vertex] (a) at (90:1) {};
    \node[vertex] (e) at (-51.43+90:1) {}
        edge (a);
    \node[vertex] (d) at (-2*51.43+90:1) {}
        edge (e);
    \node[vertex] (c) at (-3*51.43+90:1) {}
        edge (d);
    \node[vertex] (b) at (-4*51.43+90:1) {}
        edge (c);
    \node[vertex] (v6) at (-5*51.43+90:1) {}
        edge (b);
    \node[vertex] (v7) at (-6*51.43+90:1) {}
        edge (a)
        edge (v6);
    \node () at (0,-1.5) {$C_7$};

\end{scope}

\begin{scope}[xshift=8cm]

    \node[vertex] (a) at (90:1) {};
    \node[vertex] (e) at (-51.43+90:1) {};
    \node[vertex] (d) at (-2*51.43+90:1) {}
        edge (a);
    \node[vertex] (c) at (-3*51.43+90:1) {}
        edge (a)
        edge (e);
    \node[vertex] (b) at (-4*51.43+90:1) {}
        edge (a)
        edge (e)
        edge (d);
    \node[vertex] (v6) at (-5*51.43+90:1) {}
        edge (a)
        edge (e)
        edge (d)
        edge (c);
    \node[vertex] (v7) at (-6*51.43+90:1) {}
        edge (e)
        edge (d)
        edge (c)
        edge (b);

    \node () at (0,-1.5) {$\overline{C_7}$};

\end{scope}

\begin{scope}[xshift=12cm]

    \node[vertex] (a) at (90:1) {};
    \node[vertex] (e) at (-72+90:1) {}
        edge (a);
    \node[vertex] (d) at (-2*72+90:1) {}
        edge (e);
    \node[vertex] (c) at (-3*72+90:1) {}
        edge (d);
    \node[vertex] (b) at (-4*72+90:1) {}
        edge (a)
        edge (c);
    \node[vertex] (u) at (-0.375,-0.25) {}
        edge (a)
        edge (e)
        edge (d)
        edge (c)
        edge (b);
    \node[vertex] (v) at (0.375,-0.25) {}
        edge (a)
        edge (e)
        edge (d)
        edge (c)
        edge (b)
        edge (u);
    \node () at (0,-1.5) {$C_5+2$};

\end{scope}

\end{tikzpicture}
}
\caption{The ``trivial'' 7-cores. We prove in Theorem \ref{thm:Trivial7coresProjective} that they are projective.}
\label{fig:Trivial7cores}
\end{figure}
\end{center}

By Theorem \ref{thm:Classification7Cores}, there are 10 cores on 7 vertices. We call the four 7-cores of Figure \ref{fig:Trivial7cores} ``trivial'' since it is very easy to prove that they are projective.

\begin{theorem}\label{thm:Trivial7coresProjective}

The graphs $K_7,C_7,\overline{C_7}$ and $C_5+2$ presented in Figure \ref{fig:Trivial7cores} are projective cores.

\end{theorem}

\begin{proof}
The cases of $K_7,C_7,\overline{C_7}$ and $C_5+2$ have been treated, respectively, in Theorems \ref{thm:CliqueProjective}, \ref{thm:cycle_core}, \ref{thm:ComplementCycleProjectiveCore} and \ref{thm:C5+pProjective}.
\end{proof}

\begin{center}
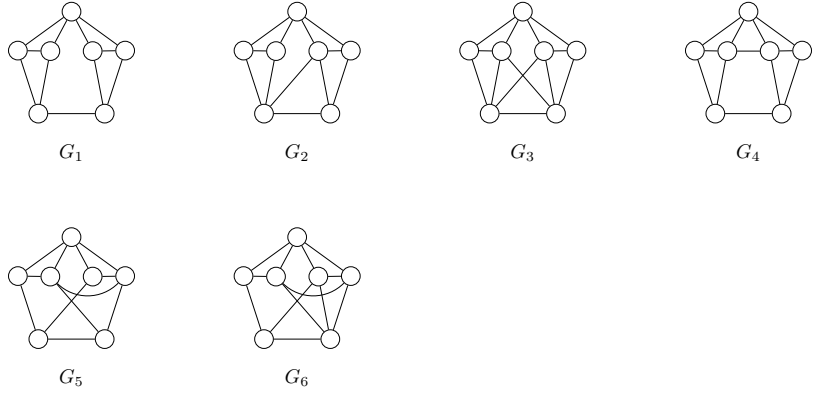
\begin{figure}
\centering
\scalebox{0.75}{
\begin{tikzpicture}

\tikzstyle{vertex}=[draw,shape=circle];

\begin{scope}[yshift=-4cm]

    \node[vertex] (v1) at (90:1) {};
    \node[vertex] (v2) at (-72+90:1) {}
        edge (v1);
    \node[vertex] (v3) at (-2*72+90:1) {}
        edge (v2);
    \node[vertex] (v4) at (-3*72+90:1) {}
        edge (v3);
    \node[vertex] (v5) at (-4*72+90:1) {}
        edge (v1)
        edge (v4);
    \node[vertex] (u) at (-0.375,0.3) {}
        edge (v1)
        edge (v4)
        edge (v5);
    \node[vertex] (v) at (0.375,0.3) {}
        edge (v1)
        edge (v2)
        edge (v3);

    \node () at (0,-1.5) {$G_1$};

\end{scope}

\begin{scope}[xshift=4cm,yshift=-4cm]

    \node[vertex] (a) at (90:1) {};
    \node[vertex] (e) at (-72+90:1) {}
        edge (a);
    \node[vertex] (d) at (-2*72+90:1) {}
        edge (e);
    \node[vertex] (c) at (-3*72+90:1) {}
        edge (d);
    \node[vertex] (b) at (-4*72+90:1) {}
        edge (a)
        edge (c);
    \node[vertex] (u) at (-0.375,0.3) {}
        edge (a)
        edge (b);
    \node[vertex] (v) at (0.375,0.3) {}
        edge (a)
        edge (e);

    \node () at (0,-1.5) {$G_2$};

\draw (c)--(u);
\draw (c)--(v);
\draw (d)--(v);

\end{scope}

\begin{scope}[xshift=8cm,yshift=-4cm]

    \node[vertex] (a) at (90:1) {};
    \node[vertex] (e) at (-72+90:1) {}
        edge (a);
    \node[vertex] (d) at (-2*72+90:1) {}
        edge (e);
    \node[vertex] (c) at (-3*72+90:1) {}
        edge (d);
    \node[vertex] (b) at (-4*72+90:1) {}
        edge (a)
        edge (c);
    \node[vertex] (u) at (-0.375,0.3) {}
        edge (a)
        edge (b);
    \node[vertex] (v) at (0.375,0.3) {}
        edge (a)
        edge (e);

    \node () at (0,-1.5) {$G_3$};

\draw (c)--(u);
\draw (c)--(v);
\draw (d)--(u);
\draw (d)--(v);

\end{scope}

\begin{scope}[xshift=12cm,yshift=-4cm]

    \node[vertex] (a) at (90:1) {};
    \node[vertex] (e) at (-72+90:1) {}
        edge (a);
    \node[vertex] (d) at (-2*72+90:1) {}
        edge (e);
    \node[vertex] (c) at (-3*72+90:1) {}
        edge (d);
    \node[vertex] (b) at (-4*72+90:1) {}
        edge (a)
        edge (c);
    \node[vertex] (u) at (-0.375,0.3) {}
        edge (a)
        edge (b);
    \node[vertex] (v) at (0.375,0.3) {}
        edge (a)
        edge (e);

    \node () at (0,-1.5) {$G_4$};

\draw (u)--(v);

\draw (c)--(u);
\draw (d)--(v);

\end{scope}

\begin{scope}[yshift=-8cm]

    \node[vertex] (a) at (90:1) {};
    \node[vertex] (e) at (-72+90:1) {}
        edge (a);
    \node[vertex] (d) at (-2*72+90:1) {}
        edge (e);
    \node[vertex] (c) at (-3*72+90:1) {}
        edge (d);
    \node[vertex] (b) at (-4*72+90:1) {}
        edge (a)
        edge (c);
    \node[vertex] (u) at (-0.375,0.3) {}
        edge (a)
        edge (b);
    \node[vertex] (v) at (0.375,0.3) {}
        edge (a)
        edge (e);

    \node () at (0,-1.5) {$G_5$};

\draw (e) to [bend left = 45] (u);

\draw (c)--(v);
\draw (d)--(u);

\end{scope}

\begin{scope}[xshift=4cm,yshift=-8cm]

    \node[vertex] (a) at (90:1) {};
    \node[vertex] (e) at (-72+90:1) {}
        edge (a);
    \node[vertex] (d) at (-2*72+90:1) {}
        edge (e);
    \node[vertex] (c) at (-3*72+90:1) {}
        edge (d);
    \node[vertex] (b) at (-4*72+90:1) {}
        edge (a)
        edge (c);
    \node[vertex] (u) at (-0.375,0.3) {}
        edge (a)
        edge (b);
    \node[vertex] (v) at (0.375,0.3) {}
        edge (a)
        edge (e);

    \node () at (0,-1.5) {$G_6$};

\draw (e) to [bend left = 45] (u);

\draw (c)--(v);
\draw (d)--(u);
\draw (d)--(v);

\end{scope}

\end{tikzpicture}
}
\caption{The ``sporadic'' 7-cores. We prove in Theorem \ref{thm:conjecture_true_7_vertices} that they are projective.}
\label{fig:Sporadic7Cores}
\end{figure}

\end{center}

What remains is now to check the projectivity of the ``sporadic'' 7-cores presented in Figure \ref{fig:Sporadic7Cores}.
To simplify this we make use of Rosenberg's classification of minimal clones~\cite{rosenberg1986minimal}, here presented in a slightly condensed form specifically for projective graphs.

\begin{theorem}\cite{rosenberg1986minimal}
Let $G$ be a non-projective graph. Then $\pol(G)$ contains a function $f$ of one of the following type:

\begin{enumerate}
    \item $f \colon (x,y,z)\mapsto x+y+z$, where $(V_G,+)$ is the additive group of a $\mathbb{F}_2$-vector space.
    \item $f$ is a ternary majority operation, i.e., $\forall (x,y)\in (V_G)^2, f(x,x,y)=f(x,y,x)=f(y,x,x)=x$.
    \item $f$ is a {\em semiprojection} of arity $m\ge 2$, i.e. $f$ is not a projection, and there exists $i\in [m]$ such that $\forall (x_1,\dots,x_m)\in (V_G)^m, |\{x_1,\dots,x_m\}|<m\implies f(x_1,\dots,x_m)=x_i$. 
\end{enumerate}

\end{theorem}

We know from the algebraic formulation of the \Csp{} dichotomy theorem (see, e.g., the survey by Barto et al.~\cite{barto2017}) that if $G$ is a graph where $\pol(G)$ contains a polymorphism of type 1 or 2, then $G$ is bipartite.
We therefore derive the following corollary.

\begin{corollary}\label{cor:semiprojection_projectivity}

Let $G$ be a core on at least 3 vertices such that $\pol(G)$ does not contain any semiprojections. Then, $G$ is a projective core.

\end{corollary}

In order to apply Corollary \ref{cor:semiprojection_projectivity} to the graphs $G_1,\dots,G_6$, we carry out a reasoning in two steps:

\begin{itemize}
    \item We verified by computer search~\cite{pcptools} that $\pol(G_1),\dots,\pol(G_6)$ do not contain any semiprojection \underline{of arity $2$ and $3$}.
    \item We prove that $\pol(G_1),\dots,\pol(G_6)$ do not contain any semiprojection \underline{of arity $\ge 4$}.
\end{itemize}

The exclusion of non-trivial semiprojections of arity $\ge 4$ is enabled by the following lemma.

\begin{lemma}\label{lem:mindegree_semiprojections}

Let $G$ be a core graph on at most $n\ge 3$ vertices, and denote by $\delta>0$ the minimal degree of a vertex in $G$. Let $m:=\lfloor \frac{n-1}{\delta} \rfloor+1$ ($m$ is an integer that satisfies $1+m\delta>n$). Then, $\pol(G)$ does not contain any semiprojection of arity $\ge m+1$.

\end{lemma}

\begin{proof}

First, note that since $G$ is a core with $G\neq K_1$, $G$ has no isolated vertex. It follows that $\delta>0$.
Assume there is a semiprojection $f\in\pol(G)$ of arity $M\ge m+1$. Thus, $f$ is not a projection.
We can assume, up to permute the coordinates, that $f$ is a semiprojection on the first coordinate, i.e.,
\[\forall (x_1,\dots,x_M)\in (V_G)^M, |\{x_1,\dots,x_M\}|<M \implies f(x_1,\dots,x_M)=x_1.\]
Since $f$ is not the projection on the first coordinate, there exists $(a_1,\dots,a_M)\in (V_G)^M$ such that $f(a_1,\dots,a_M) = a \neq a_1$.
For each vertex $u$ of $G$, let $N_G(u):=\{ v\in V_G\mid (u,v)\in E_G\}$ be the {\em open neighborhood of $u$ in $G$}.
We claim that $N_G(a_1)\setminus N_G(a)\neq \emptyset$. Indeed if we assume by contradiction that $N_G(a_1)\subseteq N_G(a)$, then the function $h:V_G\mapsto V_G$ that maps $a_1$ to $a$ and that leaves the rest of the graph unchanged would be a non-bijective (since $a_1\neq a$) $G$-coloring of $G$, contradicting the fact that $G$ is a core.
We can therefore take $x_1\in N_G(a_1)\setminus N_G(a)$.
Now, remark that, for cardinality reasons, the sets $\{x_1\},N_G(a_2),\dots,N_G(a_M)$ can not be pairwise disjoint: because they are all contained in $V_G$ and because

$$|\{x_1\}|+\sum\limits_{j=2}^M|N_G(a_j)| \ge 1+ (M-1)\delta \ge 1+m\delta > n=|V_G|.$$

We can therefore consider $(x_2,\dots,x_M) \in N_G(x_2)\times\dots\times N_G(x_M)$, such that there exists $(j_0,j_1)\in [M]^2$ with $j_0\neq j_1$ and $x_{j_0}=x_{j_1}$.

\begin{itemize}
    \item On the one hand since $x_{j_0}=x_{j_1}$, we have $|\{ x_1,\dots,x_M \}|<M$. Since $f$ is a semiprojection on the first coordinate, we deduce that $f(x_1,\dots,x_M)=x_1$.

    \item On the other hand, since we have for all $j\in [M]$ that $(a_j,x_j)^{\top}\in E_G$ (because by definition, $x_j\in N_G(a_j)$). Since $f\in\pol(G)$, we deduce that $( f(a_1,\dots,a_M) , f(x_1,\dots,x_M) )^{\top} \in E_G$. 
\end{itemize}

We obtain that $(a,x_1) \in E_G$ (recall that we defined $a:=f(a_1,\dots,a_M)$), contradicting the definition of $x_1$ (that $x_1\notin N_G(a)$).
Hence, a semiprojection of arity $\ge m+1$ cannot exist.
\end{proof}

By Lemma \ref{lem:mindegree_semiprojections}, observing that the minimal degree in the each of the core graphs $G_1,\dots,G_6$ is $3$, we deduce that there is no semiprojection in the $\pol(G_i)$ with $i\in [6]$ of arity $\ge (\lfloor \frac{7-1}{3}\rfloor+1)+1=4$.

Our main result in this theorem now follows by Theorem~\ref{thm:Classification7Cores}, Theorem~\ref{thm:Trivial7coresProjective}, and Corollary~\ref{cor:semiprojection_projectivity}.

\begin{theorem}\label{thm:conjecture_true_7_vertices}

All (indecomposable) cores on $7$ vertices are projectives.
    
\end{theorem}

Hence, we have verified the Okrasa and Rz\c ażewski Conjecture for all graphs with at most $7$ vertices, and thereby proved Theorem~\ref{thm:ConjectureTrueSevenVertices}.

\section{Conclusion and Future Research}\label{sec:conclusion}

In this paper, we have investigated the inclusion structure of the sets of partial polymorphisms of graphs, and proved that for all pairs of graphs $H$,$H'$ on the same set of vertices, $\ppol(H)\subseteq\ppol(H')$ implies that $H=H'$ or $E_{H'}=\emptyset$. Since this inclusion structure is trivial, it is natural to generalize the question and investigate inclusions of the form $\ppol(H) \subsetneq \ppol( R)$, where $H$ is a graph, but where $R$ is an arbitrary relation. We deemed the case when \Csp$(R)$ was \NP-complete to be of particular interest since the problem \Csp$(R)$ then bounds the complexity of \textsc{$H$-Coloring} from below, in a non-trivial way. We then identified a condition depending on the length of the shortest odd cycle in $H$ (the {\em odd-girth} of $H$), and proved that there exists such an $n$-ary relation $R$ if and only if the odd-girth of $H$ is $\leq n$, otherwise, \Csp$(R)$ must be trivial.
In an attempt to better understand the algebraic invariants of graphs, we then proceeded by studying total polymorphisms of graphs, with a particular focus on projective graphs, where we used the algebraic approach to obtain simplified and uniform proofs. Importantly, we used our algebraic tools to verify the Okrasa and Rz\c ażewski Conjecture for all graphs of at most 7 vertices.

Concerning future research perhaps
the most pressing question is whether we can use our algebraic results to prove (or disprove) the Okrasa and Rz\c ażewski Conjecture for graphs with more than $7$ vertices. %
By Corollary \ref{cor:projectivity}, the Okrasa and Rz\c ażewski Conjecture  is equivalent to the following statement.

\begin{conjecture}
Let $H$ be a connected core on $k\geq 3$ vertices. Then, $H$ is indecomposable if and only if $\text{NEQ}_k\in\cclone{E_H}$.
\end{conjecture}

To advance our understanding of the fine-grained complexity of \textsc{$H$-Coloring}, it would also be interesting to settle the following question.

\begin{question}
Let $H$ be a projective core. Describe $\ppol(H)$.
\end{question}

For instance, %
is it possible to relate $\ppol(H)$ with the {\em treewidth} of $H$? More generally, are there structural properties of classes of (partial) polymorphisms that translate into bounded width classes of graphs \cite{GuillemotM14}? These questions constitute topics that we are currently investigating.

\bibliographystyle{plain}
\bibliography{references}

\appendix

\section{Classification of 7-cores}\label{app:7cores}

The goal of this appendix is to prove Theorem \ref{thm:Classification7Cores}.

\ClassificationSevenCores*

\begin{center}
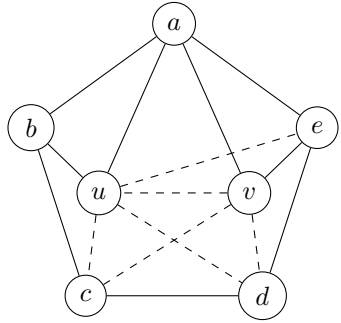
\begin{figure}
\centering

\begin{tikzpicture}
\tikzstyle{vertex}=[draw,shape=circle];

\begin{scope}

    \node[vertex] (a) at (90:2) {$a$};
    \node[vertex] (e) at (-72+90:2) {$e$}
        edge (a);
    \node[vertex] (d) at (-2*72+90:2) {$d$}
        edge (e);
    \node[vertex] (c) at (-3*72+90:2) {$c$}
        edge (d);
    \node[vertex] (b) at (-4*72+90:2) {$b$}
        edge (a)
        edge (c);
    \node[vertex] (u) at (-1,-0.25) {$u$}
        edge (a)
        edge (b)
        edge[dashed] (c)
        edge[dashed] (d)
        edge[dashed] (e);
    \node[vertex] (v) at (1,-0.25) {$v$}
        edge (a)
        edge[dashed] (c)
        edge[dashed] (d)
        edge (e)
        edge[dashed] (u);

\end{scope}

\end{tikzpicture}
\caption{Up to isomorphism, any ``sporadic'' $7$-core (not $K_7,C_7,\overline{C_7}$ or $C_5+2$) $(\{a, b, c, d, e, u, v\}, E$) must satisfy this motif: $\{a,b,c,d,e\}$ induces a $C_5$, $(b,v) \notin E$, all dashed edges are allowed (as long as $(u,v) \notin E$, or $(e,u) \notin E$), that $(c,u) \in E$ or $(c,v) \in E$, and that $(d,u) \in E$ or $(d,v) \in E$. }
\label{fig:AllPossibleGraphs}
\end{figure}
\end{center}

The high-level arguments are as follows.

\begin{enumerate}
    \item Enumerate all the graphs compatible with the motif described in Figure \ref{fig:AllPossibleGraphs}, and keep only, among these graphs, the cores.
    \item Keep exactly one representative for each class of isomorphism.
    \item Prove that all the {\em sporadic} 7-cores --- core graphs on 7 vertices that are not $K_7,C_7,\overline{C_7}$ and $C_5+2$ --- have to be compatible with the motif described in Figure~\ref{fig:AllPossibleGraphs}.
\end{enumerate}

We begin with the first step. In Figure \ref{fig:NoteuNotuv}, \ref{fig:uvNoteu} and \ref{fig:euNotuv}, we do a case analysis for all compatible graphs, depending on the 3 possible cases for $(u,v)$ and $(u,e)$. We eliminate the non-cores by proving that their core is either $K_3,K_4$ or $C_5+1$, and we both show their core as an induced subgraph (represented by the thick edges), as well as giving a homomorphism to their core, represented by the colors on the vertices. Reciprocally, we ensure that the remaining graphs are cores due to Lemma~\ref{lem:SufficientCondition7Core}.

\begin{lemma}\label{lem:SufficientCondition7Core}

Let $G$ be a graph on $7$ vertices such that:

\begin{itemize}
    \item $G$ is not $3$-colorable,
    \item $G$ has no induced $K_4$, and
    \item $G$ has no vertex of degree $\ge 5$.
\end{itemize}

Then, $G$ is a core.

\end{lemma}

\begin{proof}
Suppose, for the sake of contradiction, that $G$ is not a core. This implies that $core(G)$ must be a core graph with no more than 6 vertices. According to Remark~\ref{rem:CoresAtMost5Vertices} and Theorem \ref{thm:Cores6Vertices} in Section \ref{sec:Conjecture6Vertices}, it follows that $core(G)$ is one of the graphs in the set $\{K_n, n\in [6]\} \cup \{C_5,C_5+1\}$, as illustrated in Figure \ref{fig:CoresAtMost6Vertices}.

Given that $G$ is not 3-colorable, we can conclude that $core(G) \notin \{K_1,K_2,K_3,C_5\}$. Furthermore, the absence of an induced $K_4$ in $G$ means $core(G) \notin \{K_4,K_5,K_6\}$. Additionally, the fact that $G$ does not have any vertex with degree 5 or higher eliminates the possibility of $core(G)$ being $C_5+1$.

This leads to a contradiction: hence, $G$ must indeed be a core.
\end{proof}

We now continue with the second step, and study the isomorphisms between the obtained core graphs.
Up to isomorphism, we find 6 ``sporadic'' 7-cores:

\begin{itemize}
    \item 11 edges: only $G_1$,
    \item 12 edges: $G_2$, $G_4$, $G_5$,
    \item 13 edges: $G_3$, $G_6$.
\end{itemize}

We discuss the question of possible isomorphisms here.  Recall that isomorphic graphs always have  the same number of edges.

\begin{itemize}
    \item $G_4$ is not isomorphic to $G_2$ because if we keep only vertices of degree $4$, $G_4$ becomes a triangle and $G_2$ a $P_3$.
    \item $G_5$ is not isomorphic to $G_2$ because if we keep only vertices of degree $4$, $G_5$ becomes a triangle and $G_2$ a $P_3$.
    \item $G_5$ is not isomorphic to $G_4$ because if we keep only vertices of degree $3$, $G_4$ becomes a $P_4$ and $G_5$ a star $K_{1,3}$.
    \item $G_6$ is not isomorphic to $G_3$ because in $G_6$ the 2 vertices of degree $3$ are adjacent, and not in $G_3$.
    \item $G'_4$ (see Figure \ref{fig:euNotuv}) is isomorphic to $G_4$, by the isomorphism from $G'_4$ to $G_4$: $(a\ u\ b\ c\ d\ e)=\begin{pmatrix} a & b & c & d & e & u & v \\ u & c & d & e & a & b & v  \end{pmatrix}.$
    
    \item $G'_3$ (see Figure \ref{fig:euNotuv}) is isomorphic to $G_3$ through the isomorphism from $G'_3$ to $G_3$: $(a\ d\ b\ e\ c)(u\ v)=\begin{pmatrix} a & b & c & d & e & u & v \\ d & e & a & b & c & v & u \end{pmatrix}.$
\end{itemize}

\begin{center}
\begin{figure}[p]
\centering
\scalebox{0.8}{
\begin{tikzpicture}

\tikzstyle{vertex}=[draw,shape=circle];

\begin{scope}[yshift=-4cm]

    \node () at (0,0) {$(d,u)\in E, (d,v)\notin E$};

\end{scope}

\begin{scope}[yshift=-8cm]

    \node () at (0,0) {$(d,u)\notin E, (d,v)\in E$};

\end{scope}

\begin{scope}[yshift=-12cm]

    \node () at (0,0) {$(d,u)\in E, (d,v)\in E$};

\end{scope}

\begin{scope}[xshift=4cm,yshift=-2cm]

    \node () at (0,0) {$(c,u)\in E, (c,v)\notin E$};

\end{scope}

\begin{scope}[xshift=8cm,yshift=-2cm]

    \node () at (0,0) {$(c,u)\notin E, (c,v)\in E$};

\end{scope}

\begin{scope}[xshift=12cm,yshift=-2cm]

    \node () at (0,0) {$(c,u)\in E, (c,v)\in E$};

\end{scope}

\begin{scope}[xshift=4cm,yshift=-4cm]

    \node[vertex,fill=blue] (a) at (90:1) {};
    \node[vertex,fill=red] (e) at (-72+90:1) {}
        edge (a);
    \node[vertex,fill=green] (d) at (-2*72+90:1) {}
        edge (e);
    \node[vertex,fill=blue] (c) at (-3*72+90:1) {}
        edge[ultra thick] (d);
    \node[vertex,fill=green] (b) at (-4*72+90:1) {}
        edge (a)
        edge (c);
    \node[vertex,fill=red] (u) at (-0.375,0.3) {}
        edge (a)
        edge (b);
    \node[vertex,fill=green] (v) at (0.375,0.3) {}
        edge (a)
        edge (e);

    \node () at (0,-1.5) {its core is $K_3$};

\draw[ultra thick] (c)--(u);
\draw[ultra thick] (d)--(u);

\end{scope}

\begin{scope}[xshift=8cm,yshift=-4cm]

    \node[vertex,fill=blue] (a) at (90:1) {};
    \node[vertex,fill=red] (e) at (-72+90:1) {}
        edge (a);
    \node[vertex,fill=green] (d) at (-2*72+90:1) {}
        edge (e);
    \node[vertex,fill=blue] (c) at (-3*72+90:1) {}
        edge (d);
    \node[vertex,fill=green] (b) at (-4*72+90:1) {}
        edge[ultra thick] (a)
        edge (c);
    \node[vertex,fill=red] (u) at (-0.375,0.3) {}
        edge[ultra thick] (a)
        edge[ultra thick] (b);
    \node[vertex,fill=green] (v) at (0.375,0.3) {}
        edge (a)
        edge (e);

    \node () at (0,-1.5) {its core is $K_3$};

\draw (c)--(v);
\draw (d)--(u);

\end{scope}

\begin{scope}[xshift=12cm,yshift=-4cm]

    \node[vertex,fill=blue] (a) at (90:1) {};
    \node[vertex,fill=red] (e) at (-72+90:1) {}
        edge (a);
    \node[vertex,fill=green] (d) at (-2*72+90:1) {}
        edge (e);
    \node[vertex,fill=blue] (c) at (-3*72+90:1) {}
        edge[ultra thick] (d);
    \node[vertex,fill=green] (b) at (-4*72+90:1) {}
        edge (a)
        edge (c);
    \node[vertex,fill=red] (u) at (-0.375,0.3) {}
        edge (a)
        edge (b);
    \node[vertex,fill=green] (v) at (0.375,0.3) {}
        edge (a)
        edge (e);

    \node () at (0,-1.5) {its core is $K_3$};

\draw[ultra thick] (c)--(u);
\draw (c)--(v);
\draw[ultra thick] (d)--(u);

\end{scope}

\begin{scope}[xshift=4cm,yshift=-8cm]

    \node[vertex] (a) at (90:1) {};
    \node[vertex] (e) at (-72+90:1) {}
        edge (a);
    \node[vertex] (d) at (-2*72+90:1) {}
        edge (e);
    \node[vertex] (c) at (-3*72+90:1) {}
        edge (d);
    \node[vertex] (b) at (-4*72+90:1) {}
        edge (a)
        edge (c);
    \node[vertex] (u) at (-0.375,0.3) {}
        edge (a)
        edge (b);
    \node[vertex] (v) at (0.375,0.3) {}
        edge (a)
        edge (e);

    \node () at (0,-1.5) {\textbf{Core} $G_1$};

\draw (c)--(u);
\draw (d)--(v);

\draw[ultra thick] (-1.5,1.5)--(1.5,1.5)--(1.5,-2)--(-1.5,-2)--(-1.5,1.5);

\end{scope}

\begin{scope}[xshift=8cm,yshift=-8cm]

    \node[vertex,fill=blue] (a) at (90:1) {};
    \node[vertex,fill=green] (e) at (-72+90:1) {}
        edge (a);
    \node[vertex,fill=blue] (d) at (-2*72+90:1) {}
        edge (e);
    \node[vertex,fill=green] (c) at (-3*72+90:1) {}
        edge[ultra thick] (d);
    \node[vertex,fill=red] (b) at (-4*72+90:1) {}
        edge (a)
        edge (c);
    \node[vertex,fill=green] (u) at (-0.375,0.3) {}
        edge (a)
        edge (b);
    \node[vertex,fill=red] (v) at (0.375,0.3) {}
        edge (a)
        edge (e);

    \node () at (0,-1.5) {its core is $K_3$};

\draw[ultra thick] (c)--(v);
\draw[ultra thick] (d)--(v);

\end{scope}

\begin{scope}[xshift=12cm,yshift=-8cm]

    \node[vertex] (a) at (90:1) {};
    \node[vertex] (e) at (-72+90:1) {}
        edge (a);
    \node[vertex] (d) at (-2*72+90:1) {}
        edge (e);
    \node[vertex] (c) at (-3*72+90:1) {}
        edge (d);
    \node[vertex] (b) at (-4*72+90:1) {}
        edge (a)
        edge (c);
    \node[vertex] (u) at (-0.375,0.3) {}
        edge (a)
        edge (b);
    \node[vertex] (v) at (0.375,0.3) {}
        edge (a)
        edge (e);

    \node () at (0,-1.5) {\textbf{Core} $G_2$};

\draw (c)--(u);
\draw (c)--(v);
\draw (d)--(v);

\draw[ultra thick] (-1.5,1.5)--(1.5,1.5)--(1.5,-2)--(-1.5,-2)--(-1.5,1.5);

\end{scope}

\begin{scope}[xshift=4cm,yshift=-12cm]

    \node[vertex] (a) at (90:1) {};
    \node[vertex] (e) at (-72+90:1) {}
        edge (a);
    \node[vertex] (d) at (-2*72+90:1) {}
        edge (e);
    \node[vertex] (c) at (-3*72+90:1) {}
        edge (d);
    \node[vertex] (b) at (-4*72+90:1) {}
        edge (a)
        edge (c);
    \node[vertex] (u) at (-0.375,0.3) {}
        edge (a)
        edge (b);
    \node[vertex] (v) at (0.375,0.3) {}
        edge (a)
        edge (e);

    \node () at (0,-1.5) {isomorphic to $G_2$};

\draw (c)--(u);
\draw (d)--(u);
\draw (d)--(v);

\end{scope}

\begin{scope}[xshift=8cm,yshift=-12cm]

    \node[vertex,fill=blue] (a) at (90:1) {};
    \node[vertex,fill=green] (e) at (-72+90:1) {}
        edge (a);
    \node[vertex,fill=blue] (d) at (-2*72+90:1) {}
        edge (e);
    \node[vertex,fill=green] (c) at (-3*72+90:1) {}
        edge[ultra thick] (d);
    \node[vertex,fill=red] (b) at (-4*72+90:1) {}
        edge (a)
        edge (c);
    \node[vertex,fill=green] (u) at (-0.375,0.3) {}
        edge (a)
        edge (b);
    \node[vertex,fill=red] (v) at (0.375,0.3) {}
        edge (a)
        edge (e);

    \node () at (0,-1.5) {its core is $K_3$};

\draw[ultra thick] (c)--(v);
\draw (d)--(u);
\draw[ultra thick] (d)--(v);

\end{scope}

\begin{scope}[xshift=12cm,yshift=-12cm]

    \node[vertex] (a) at (90:1) {};
    \node[vertex] (e) at (-72+90:1) {}
        edge (a);
    \node[vertex] (d) at (-2*72+90:1) {}
        edge (e);
    \node[vertex] (c) at (-3*72+90:1) {}
        edge (d);
    \node[vertex] (b) at (-4*72+90:1) {}
        edge (a)
        edge (c);
    \node[vertex] (u) at (-0.375,0.3) {}
        edge (a)
        edge (b);
    \node[vertex] (v) at (0.375,0.3) {}
        edge (a)
        edge (e);

    \node () at (0,-1.5) {\textbf{Core} $G_3$};

\draw (c)--(u);
\draw (c)--(v);
\draw (d)--(u);
\draw (d)--(v);

\draw[ultra thick] (-1.5,1.5)--(1.5,1.5)--(1.5,-2)--(-1.5,-2)--(-1.5,1.5);

\end{scope}

\draw[ultra thick] (-2,-2.5)--(14,-2.5);
\draw[ultra thick] (2,-1.5)--(2,-14);

\end{tikzpicture}
}
\caption{The $9$ candidates with $(e,u)\notin E$ and $(u,v)\notin E$.}
\label{fig:NoteuNotuv}
\end{figure}
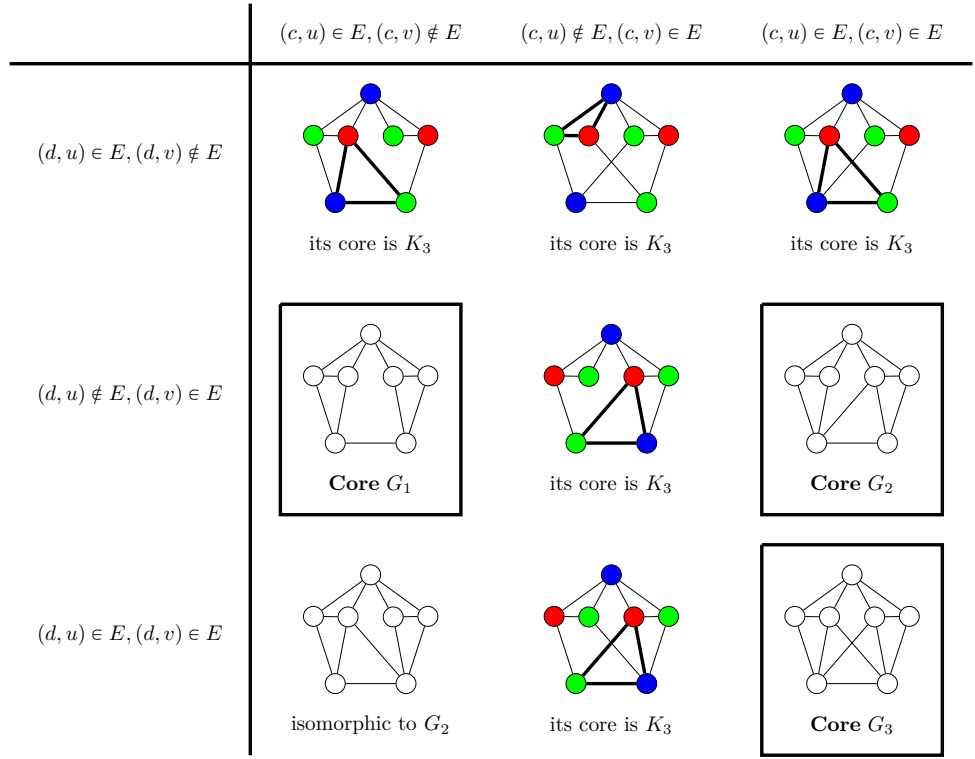
\end{center}

\begin{center}
\begin{figure}[p]
\centering
\scalebox{0.8}{
\begin{tikzpicture}

\tikzstyle{vertex}=[draw,shape=circle];

\begin{scope}[yshift=-4cm]

    \node () at (0,0) {$(d,u)\in E, (d,v)\notin E$};

\end{scope}

\begin{scope}[yshift=-8cm]

    \node () at (0,0) {$(d,u)\notin E, (d,v)\in E$};

\end{scope}

\begin{scope}[yshift=-12cm]

    \node () at (0,0) {$(d,u)\in E, (d,v)\in E$};

\end{scope}

\begin{scope}[xshift=4cm,yshift=-2cm]

    \node () at (0,0) {$(c,u)\in E, (c,v)\notin E$};

\end{scope}

\begin{scope}[xshift=8cm,yshift=-2cm]

    \node () at (0,0) {$(c,u)\notin E, (c,v)\in E$};

\end{scope}

\begin{scope}[xshift=12cm,yshift=-2cm]

    \node () at (0,0) {$(c,u)\in E, (c,v)\in E$};

\end{scope}

\begin{scope}[xshift=4cm,yshift=-4cm]

    \node[vertex,fill=blue] (a) at (90:1) {};
    \node[vertex,fill=red] (e) at (-72+90:1) {}
        edge (a);
    \node[vertex,fill=green] (d) at (-2*72+90:1) {}
        edge (e);
    \node[vertex,fill=blue] (c) at (-3*72+90:1) {}
        edge[ultra thick] (d);
    \node[vertex,fill=green] (b) at (-4*72+90:1) {}
        edge (a)
        edge (c);
    \node[vertex,fill=red] (u) at (-0.375,0.3) {}
        edge (a)
        edge (b);
    \node[vertex,fill=green] (v) at (0.375,0.3) {}
        edge (a)
        edge (e);

    \node () at (0,-1.5) {its core is $K_3$};

\draw (u)--(v);

\draw[ultra thick] (c)--(u);
\draw[ultra thick] (d)--(u);

\end{scope}

\begin{scope}[xshift=8cm,yshift=-4cm]

    \node[vertex,fill=blue] (a) at (90:1) {};
    \node[vertex,fill=red] (e) at (-72+90:1) {}
        edge (a);
    \node[vertex,fill=green] (d) at (-2*72+90:1) {}
        edge (e);
    \node[vertex,fill=blue] (c) at (-3*72+90:1) {}
        edge (d);
    \node[vertex,fill=green] (b) at (-4*72+90:1) {}
        edge (a)
        edge (c);
    \node[vertex,fill=red] (u) at (-0.375,0.3) {}
        edge[ultra thick] (a)
        edge (b);
    \node[vertex,fill=green] (v) at (0.375,0.3) {}
        edge[ultra thick] (a)
        edge (e);

    \node () at (0,-1.5) {its core is $K_3$};

\draw[ultra thick] (u)--(v);

\draw (c)--(v);
\draw (d)--(u);

\end{scope}

\begin{scope}[xshift=12cm,yshift=-4cm]

    \node[vertex,fill=blue] (a) at (90:1) {};
    \node[vertex,fill=red] (e) at (-72+90:1) {}
        edge (a);
    \node[vertex,fill=green] (d) at (-2*72+90:1) {}
        edge (e);
    \node[vertex,fill=blue] (c) at (-3*72+90:1) {}
        edge[ultra thick] (d);
    \node[vertex,fill=green] (b) at (-4*72+90:1) {}
        edge (a)
        edge (c);
    \node[vertex,fill=red] (u) at (-0.375,0.3) {}
        edge (a)
        edge (b);
    \node[vertex,fill=green] (v) at (0.375,0.3) {}
        edge (a)
        edge (e);

    \node () at (0,-1.5) {its core is $K_3$};

\draw (u)--(v);

\draw[ultra thick] (c)--(u);
\draw (c)--(v);
\draw[ultra thick] (d)--(u);

\end{scope}

\begin{scope}[xshift=4cm,yshift=-8cm]

    \node[vertex] (a) at (90:1) {};
    \node[vertex] (e) at (-72+90:1) {}
        edge (a);
    \node[vertex] (d) at (-2*72+90:1) {}
        edge (e);
    \node[vertex] (c) at (-3*72+90:1) {}
        edge (d);
    \node[vertex] (b) at (-4*72+90:1) {}
        edge (a)
        edge (c);
    \node[vertex] (u) at (-0.375,0.3) {}
        edge (a)
        edge (b);
    \node[vertex] (v) at (0.375,0.3) {}
        edge (a)
        edge (e);

    \node () at (0,-1.5) {\textbf{Core} $G_4$};

\draw (u)--(v);

\draw (c)--(u);
\draw (d)--(v);

\draw[ultra thick] (-1.5,1.5)--(1.5,1.5)--(1.5,-2)--(-1.5,-2)--(-1.5,1.5);

\end{scope}

\begin{scope}[xshift=8cm,yshift=-8cm]

    \node[vertex,fill=blue] (a) at (90:1) {};
    \node[vertex,fill=green] (e) at (-72+90:1) {}
        edge (a);
    \node[vertex,fill=blue] (d) at (-2*72+90:1) {}
        edge (e);
    \node[vertex,fill=green] (c) at (-3*72+90:1) {}
        edge[ultra thick] (d);
    \node[vertex,fill=red] (b) at (-4*72+90:1) {}
        edge (a)
        edge (c);
    \node[vertex,fill=green] (u) at (-0.375,0.3) {}
        edge (a)
        edge (b);
    \node[vertex,fill=red] (v) at (0.375,0.3) {}
        edge (a)
        edge (e);

    \node () at (0,-1.5) {its core is $K_3$};

\draw (u)--(v);

\draw[ultra thick] (c)--(v);
\draw[ultra thick] (d)--(v);

\end{scope}

\begin{scope}[xshift=12cm,yshift=-8cm]

    \node[vertex,fill=violet] (a) at (90:1) {};
    \node[vertex,fill=magenta] (e) at (-72+90:1) {}
        edge[ultra thick] (a);
    \node[vertex,fill=pink] (d) at (-2*72+90:1) {}
        edge[ultra thick] (e);
    \node[vertex,fill=teal] (c) at (-3*72+90:1) {}
        edge[ultra thick] (d);
    \node[vertex,fill=yellow] (b) at (-4*72+90:1) {}
        edge (a)
        edge (c);
    \node[vertex,fill=purple] (u) at (-0.375,0.3) {}
        edge[ultra thick] (a)
        edge (b);
    \node[vertex,fill=yellow] (v) at (0.375,0.3) {}
        edge[ultra thick] (a)
        edge[ultra thick] (e);

    \node () at (0,-1.5) {its core is $C_5+1$};

\draw[ultra thick] (u)--(v);

\draw[ultra thick] (c)--(u);
\draw[ultra thick] (c)--(v);
\draw[ultra thick] (d)--(v);

\end{scope}

\begin{scope}[xshift=4cm,yshift=-12cm]

    \node[vertex,fill=violet] (a) at (90:1) {};
    \node[vertex,fill=yellow] (e) at (-72+90:1) {}
        edge (a);
    \node[vertex,fill=pink] (d) at (-2*72+90:1) {}
        edge (e);
    \node[vertex,fill=teal] (c) at (-3*72+90:1) {}
        edge[ultra thick] (d);
    \node[vertex,fill=purple] (b) at (-4*72+90:1) {}
        edge[ultra thick] (a)
        edge[ultra thick] (c);
    \node[vertex,fill=yellow] (u) at (-0.375,0.3) {}
        edge[ultra thick] (a)
        edge[ultra thick] (b);
    \node[vertex,fill=magenta] (v) at (0.375,0.3) {}
        edge[ultra thick] (a)
        edge (e);

    \node () at (0,-1.5) {its core is $C_5+1$};

\draw[ultra thick] (u)--(v);

\draw[ultra thick] (c)--(u);
\draw[ultra thick] (d)--(u);
\draw[ultra thick] (d)--(v);

\end{scope}

\begin{scope}[xshift=8cm,yshift=-12cm]

    \node[vertex,fill=blue] (a) at (90:1) {};
    \node[vertex,fill=green] (e) at (-72+90:1) {}
        edge (a);
    \node[vertex,fill=blue] (d) at (-2*72+90:1) {}
        edge (e);
    \node[vertex,fill=green] (c) at (-3*72+90:1) {}
        edge[ultra thick] (d);
    \node[vertex,fill=red] (b) at (-4*72+90:1) {}
        edge (a)
        edge (c);
    \node[vertex,fill=green] (u) at (-0.375,0.3) {}
        edge (a)
        edge (b);
    \node[vertex,fill=red] (v) at (0.375,0.3) {}
        edge (a)
        edge (e);

    \node () at (0,-1.5) {its core is $K_3$};

\draw (u)--(v);

\draw[ultra thick] (c)--(v);
\draw (d)--(u);
\draw[ultra thick] (d)--(v);

\end{scope}

\begin{scope}[xshift=12cm,yshift=-12cm]

    \node[vertex,fill=blue] (a) at (90:1) {};
    \node[vertex,fill=green] (e) at (-72+90:1) {}
        edge (a);
    \node[vertex,fill=blue] (d) at (-2*72+90:1) {}
        edge (e);
    \node[vertex,fill=green] (c) at (-3*72+90:1) {}
        edge[ultra thick] (d);
    \node[vertex,fill=black] (b) at (-4*72+90:1) {}
        edge (a)
        edge (c);
    \node[vertex,fill=red] (u) at (-0.375,0.3) {}
        edge (a)
        edge (b);
    \node[vertex,fill=black] (v) at (0.375,0.3) {}
        edge (a)
        edge (e);

    \node () at (0,-1.5) {its core is $K_4$};

\draw[ultra thick] (u)--(v);

\draw[ultra thick] (c)--(u);
\draw[ultra thick] (c)--(v);
\draw[ultra thick] (d)--(u);
\draw[ultra thick] (d)--(v);

\end{scope}

\draw[ultra thick] (-2,-2.5)--(14,-2.5);
\draw[ultra thick] (2,-1.5)--(2,-14);

\end{tikzpicture}
}
\caption{The $9$ candidates with $(e,u)\notin E$ and $(u,v)\in E$.}
\label{fig:uvNoteu}
\end{figure}
\end{center}

\begin{center}
\begin{figure}[p]
\centering
\scalebox{0.8}{
\begin{tikzpicture}

\tikzstyle{vertex}=[draw,shape=circle];

\begin{scope}[yshift=-4cm]

    \node () at (0,0) {$(d,u)\in E, (d,v)\notin E$};

\end{scope}

\begin{scope}[yshift=-8cm]

    \node () at (0,0) {$(d,u)\notin E, (d,v)\in E$};

\end{scope}

\begin{scope}[yshift=-12cm]

    \node () at (0,0) {$(d,u)\in E, (d,v)\in E$};

\end{scope}

\begin{scope}[xshift=4cm,yshift=-2cm]

    \node () at (0,0) {$(c,u)\in E, (c,v)\notin E$};

\end{scope}

\begin{scope}[xshift=8cm,yshift=-2cm]

    \node () at (0,0) {$(c,u)\notin E, (c,v)\in E$};

\end{scope}

\begin{scope}[xshift=12cm,yshift=-2cm]

    \node () at (0,0) {$(c,u)\in E, (c,v)\in E$};

\end{scope}

\begin{scope}[xshift=4cm,yshift=-4cm]

    \node[vertex,fill=violet] (a) at (90:1) {};
    \node[vertex,fill=magenta] (e) at (-72+90:1) {}
        edge[ultra thick] (a);
    \node[vertex,fill=pink] (d) at (-2*72+90:1) {}
        edge[ultra thick] (e);
    \node[vertex,fill=teal] (c) at (-3*72+90:1) {}
        edge[ultra thick] (d);
    \node[vertex,fill=purple] (b) at (-4*72+90:1) {}
        edge[ultra thick] (a)
        edge[ultra thick] (c);
    \node[vertex,fill=yellow] (u) at (-0.375,0.3) {}
        edge[ultra thick] (a)
        edge[ultra thick] (b);
    \node[vertex,fill=yellow] (v) at (0.375,0.3) {}
        edge (a)
        edge (e);

    \node () at (0,-1.5) {its core is $C_5+1$};

\draw[ultra thick] (e) to [bend left = 45] (u);

\draw[ultra thick] (c)--(u);
\draw[ultra thick] (d)--(u);

\end{scope}

\begin{scope}[xshift=8cm,yshift=-4cm]

    \node[vertex] (a) at (90:1) {};
    \node[vertex] (e) at (-72+90:1) {}
        edge (a);
    \node[vertex] (d) at (-2*72+90:1) {}
        edge (e);
    \node[vertex] (c) at (-3*72+90:1) {}
        edge (d);
    \node[vertex] (b) at (-4*72+90:1) {}
        edge (a)
        edge (c);
    \node[vertex] (u) at (-0.375,0.3) {}
        edge (a)
        edge (b);
    \node[vertex] (v) at (0.375,0.3) {}
        edge (a)
        edge (e);

    \node () at (0,-1.5) {\textbf{Core} $G_5$};

\draw (e) to [bend left = 45] (u);

\draw (c)--(v);
\draw (d)--(u);

\draw[ultra thick] (-1.5,1.5)--(1.5,1.5)--(1.5,-2)--(-1.5,-2)--(-1.5,1.5);

\end{scope}

\begin{scope}[xshift=12cm,yshift=-4cm]

    \node[vertex,fill=violet] (a) at (90:1) {};
    \node[vertex,fill=magenta] (e) at (-72+90:1) {}
        edge[ultra thick] (a);
    \node[vertex,fill=pink] (d) at (-2*72+90:1) {}
        edge[ultra thick] (e);
    \node[vertex,fill=teal] (c) at (-3*72+90:1) {}
        edge[ultra thick] (d);
    \node[vertex,fill=purple] (b) at (-4*72+90:1) {}
        edge[ultra thick] (a)
        edge[ultra thick] (c);
    \node[vertex,fill=yellow] (u) at (-0.375,0.3) {}
        edge[ultra thick] (a)
        edge[ultra thick] (b);
    \node[vertex,fill=yellow] (v) at (0.375,0.3) {}
        edge (a)
        edge (e);

    \node () at (0,-1.5) {its core is $C_5+1$};

\draw[ultra thick] (e) to [bend left = 45] (u);

\draw[ultra thick] (c)--(u);
\draw (c)--(v);
\draw[ultra thick] (d)--(u);

\end{scope}

\begin{scope}[xshift=4cm,yshift=-8cm]

    \node[vertex] (a) at (90:1) {};
    \node[vertex] (e) at (-72+90:1) {}
        edge (a);
    \node[vertex] (d) at (-2*72+90:1) {}
        edge (e);
    \node[vertex] (c) at (-3*72+90:1) {}
        edge (d);
    \node[vertex] (b) at (-4*72+90:1) {}
        edge (a)
        edge (c);
    \node[vertex] (u) at (-0.375,0.3) {}
        edge (a)
        edge (b);
    \node[vertex] (v) at (0.375,0.3) {}
        edge (a)
        edge (e);

    \node () at (0,-1.5) {isomorphic to $G_4$};

\draw (e) to [bend left = 45] (u);

\draw (c)--(u);
\draw (d)--(v);

\end{scope}

\begin{scope}[xshift=8cm,yshift=-8cm]

    \node[vertex] (a) at (90:1) {};
    \node[vertex] (e) at (-72+90:1) {}
        edge (a);
    \node[vertex] (d) at (-2*72+90:1) {}
        edge (e);
    \node[vertex] (c) at (-3*72+90:1) {}
        edge (d);
    \node[vertex] (b) at (-4*72+90:1) {}
        edge (a)
        edge (c);
    \node[vertex] (u) at (-0.375,0.3) {}
        edge (a)
        edge (b);
    \node[vertex] (v) at (0.375,0.3) {}
        edge (a)
        edge (e);

    \node () at (0,-1.5) {$G'_4$ isomorphic to $G_4$};

\draw (e) to [bend left = 45] (u);

\draw (c)--(v);
\draw (d)--(v);

\end{scope}

\begin{scope}[xshift=12cm,yshift=-8cm]

    \node[vertex] (a) at (90:1) {};
    \node[vertex] (e) at (-72+90:1) {}
        edge (a);
    \node[vertex] (d) at (-2*72+90:1) {}
        edge (e);
    \node[vertex] (c) at (-3*72+90:1) {}
        edge (d);
    \node[vertex] (b) at (-4*72+90:1) {}
        edge (a)
        edge (c);
    \node[vertex] (u) at (-0.375,0.3) {}
        edge (a)
        edge (b);
    \node[vertex] (v) at (0.375,0.3) {}
        edge (a)
        edge (e);

    \node () at (0,-1.5) {$G'_3$ isomorphic to $G_3$};

\draw (e) to [bend left = 45] (u);

\draw (c)--(u);
\draw (c)--(v);
\draw (d)--(v);

\end{scope}

\begin{scope}[xshift=4cm,yshift=-12cm]

    \node[vertex,fill=violet] (a) at (90:1) {};
    \node[vertex,fill=magenta] (e) at (-72+90:1) {}
        edge[ultra thick] (a);
    \node[vertex,fill=pink] (d) at (-2*72+90:1) {}
        edge[ultra thick] (e);
    \node[vertex,fill=teal] (c) at (-3*72+90:1) {}
        edge[ultra thick] (d);
    \node[vertex,fill=purple] (b) at (-4*72+90:1) {}
        edge[ultra thick] (a)
        edge[ultra thick] (c);
    \node[vertex,fill=yellow] (u) at (-0.375,0.3) {}
        edge[ultra thick] (a)
        edge[ultra thick] (b);
    \node[vertex,fill=yellow] (v) at (0.375,0.3) {}
        edge (a)
        edge (e);

    \node () at (0,-1.5) {its core is $C_5+1$};

\draw[ultra thick] (e) to [bend left = 45] (u);

\draw[ultra thick] (c)--(u);
\draw[ultra thick] (d)--(u);
\draw (d)--(v);

\end{scope}

\begin{scope}[xshift=8cm,yshift=-12cm]

    \node[vertex] (a) at (90:1) {};
    \node[vertex] (e) at (-72+90:1) {}
        edge (a);
    \node[vertex] (d) at (-2*72+90:1) {}
        edge (e);
    \node[vertex] (c) at (-3*72+90:1) {}
        edge (d);
    \node[vertex] (b) at (-4*72+90:1) {}
        edge (a)
        edge (c);
    \node[vertex] (u) at (-0.375,0.3) {}
        edge (a)
        edge (b);
    \node[vertex] (v) at (0.375,0.3) {}
        edge (a)
        edge (e);

    \node () at (0,-1.5) {\textbf{Core} $G_6$};

\draw (e) to [bend left = 45] (u);

\draw (c)--(v);
\draw (d)--(u);
\draw (d)--(v);

\draw[ultra thick] (-1.5,1.5)--(1.5,1.5)--(1.5,-2)--(-1.5,-2)--(-1.5,1.5);

\end{scope}

\begin{scope}[xshift=12cm,yshift=-12cm]

    \node[vertex,fill=violet] (a) at (90:1) {};
    \node[vertex,fill=magenta] (e) at (-72+90:1) {}
        edge[ultra thick] (a);
    \node[vertex,fill=pink] (d) at (-2*72+90:1) {}
        edge[ultra thick] (e);
    \node[vertex,fill=teal] (c) at (-3*72+90:1) {}
        edge[ultra thick] (d);
    \node[vertex,fill=purple] (b) at (-4*72+90:1) {}
        edge[ultra thick] (a)
        edge[ultra thick] (c);
    \node[vertex,fill=yellow] (u) at (-0.375,0.3) {}
        edge[ultra thick] (a)
        edge[ultra thick] (b);
    \node[vertex,fill=yellow] (v) at (0.375,0.3) {}
        edge (a)
        edge (e);

    \node () at (0,-1.5) {its core is $C_5+1$};

\draw[ultra thick] (e) to [bend left = 45] (u);

\draw[ultra thick] (c)--(u);
\draw (c)--(v);
\draw[ultra thick] (d)--(u);
\draw (d)--(v);

\end{scope}

\draw[ultra thick] (-2,-2.5)--(14,-2.5);
\draw[ultra thick] (2,-1.5)--(2,-14);

\end{tikzpicture}
}
\caption{The $9$ candidates with $(e,u)\in E$ and $(u,v)\notin E$.}
\label{fig:euNotuv}
\end{figure}
\end{center}

The rest of this appendix focuses on the third step,  i.e., proving that all the ``sporadic'' 7-cores (the 7-cores that are not $K_7,C_7,\overline{C_7}$, or $C_5+2$) are necessarily compatible with the motif described in Figure \ref{fig:AllPossibleGraphs}. This is sufficient to prove that the list of $7$-cores given by Figures \ref{fig:Trivial7cores} and \ref{fig:Sporadic7Cores} is exhaustive.

\begin{lemma}\label{lem:mindegreecore}

Let $C$ a core on at least $3$ vertices, and let $v\in V_C$. Then $\text{deg}_C(v)\ge 2$.

\end{lemma}

\begin{proof}

We give a proof by contradiction by distinguishing between the following $3$ cases.

\begin{itemize}

    \item If $\text{deg}_C(v)=0$, the function that maps $v$ to any other vertex and that leaves the rest of the graph unchanged is a non-bijective homomorphism. This contradicts  that $C$ is a core.
    
    \item If $v$ has a unique neighbor $u$, and if $u$ has another neighbor $w$, the function that maps $v$ to $w$ and that leaves the rest of the graph unchanged is a non-bijective homomorphism. Again, this contradicts that $C$ is a core.

    \item If $v$ has a unique neighbor $u$, and if $u$ has no other neighbor, then, since $C$ has at least $3$ vertices, there exists $w\in V_C\setminus (u,v)$. Since, by what precedes, $\text{deg}_C(w)\neq 0$, $w$ has a neighbor $x$. The function that maps $v$ to $w$ and $u$ to $x$ and that leaves the rest of the graph unchanged is a non-bijective homomorphism. Hence, this contradicts that $C$ is a core.
\end{itemize}
\end{proof}

We now continue by establishing the following necessary properties of sporadic $7$-cores. In the following statements we implicitly assume that $G$ is a sporadic 7-core with vertices named as in Figure~\ref{fig:AllPossibleGraphs}.

\begin{lemma}\label{lem:inducedC5}

$G$ has an induced $C_5$.

\end{lemma}

\begin{proof}
Let $G$ be a sporadic $7$-core. Assume by contradiction that $G$ has no induced $C_5$. Since $G$ has $7$ vertices, $G$ has no induced $C_7$, nor has it an induced $\overline{C_7}$ (otherwise, we would have $G=C_7$ or $G=\overline{C_7}$). Thus, $G$ has no induced $C_{2k+1}$ nor  has it an induced $\overline{C_{2k+1}}$ for any $k\ge 2$ (notice that $\overline{C_5}=C_5$). By the theorem of perfect graphs, $G$ is a perfect graph, i.e., there exists $k\ge 1$ such that $G$ has an induced $K_k$ and $G$ is $k$-colorable. In particular, $core(G)=K_k$ is a clique. Since $G$ is a core, $G=core(G)=K_7$, leading to a contradiction. 
\end{proof}

By Lemma \ref{lem:inducedC5}, we can assume without loss of generality that $V_G=\{a,b,c,d,e,u,v\}$ and that $\{a,b,c,d,e\}$ induce the $C_5$: $a-b-c-d-e-a$.

Then, $G$ depends only on the neighborhoods of $u$ and $v$.

\begin{lemma}\label{lem:CommonNeighbor}

$u$ and $v$ have a common neighbor.

\end{lemma}

\begin{proof}

Assume, with the aim of reaching a contradiction, that $u$ and $v$ do not have a common neighbor, and
assume by symmetry that $\text{deg}_G(u)\ge \text{deg}_G(v)$.
By Lemma \ref{lem:mindegreecore}, $\text{deg}_G(v)\ge 2$. Consider the following case analysis.

\begin{itemize}
    \item If $G$ has no triangle, then the neighbors of $u$ (respectively $v$) are non-adjacent, and $u$ and $v$ do not have a common neighbor. It follows that $G$ is isomorphic to a subgraph of the graph presented in Figure \ref{fig:NoCommonNeighborsNoTriangle}. Thus there exists a homomorphism from $G$ to $C_5$. Since $G$ has an induced $C_5$ by Lemma \ref{lem:inducedC5}, we deduce that $core(G)=C_5$, contradicting that $G$ is a core.

\begin{center}
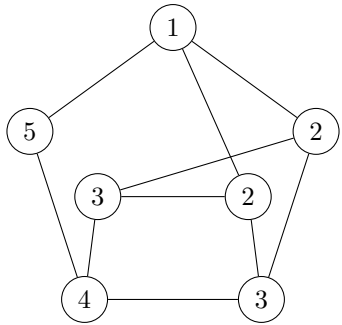
\begin{figure}
\centering

\begin{tikzpicture}
\tikzstyle{vertex}=[draw,shape=circle];

\begin{scope}

    \node[vertex] (a) at (90:2) {$1$};
    \node[vertex] (e) at (-72+90:2) {$2$}
        edge (a);
    \node[vertex] (d) at (-2*72+90:2) {$3$}
        edge (e);
    \node[vertex] (c) at (-3*72+90:2) {$4$}
        edge (d);
    \node[vertex] (b) at (-4*72+90:2) {$5$}
        edge (a)
        edge (c);
    \node[vertex] (u) at (-1,-0.25) {$3$}
        edge (e)
        edge (c);
    \node[vertex] (v) at (1,-0.25) {$2$}
        edge (a)
        edge (d)
        edge (u);

\end{scope}

\end{tikzpicture}
\caption{Maximal case (up to isomorphism) where $u$ and $v$ do not have a common neighbor and $G$ has no triangle. Even this maximal case is $C_5$-colorable.}
\label{fig:NoCommonNeighborsNoTriangle}
\end{figure}
\end{center}

    \item If $G$ has a triangle, then $G$ is isomorphic to a subgraph of one of the three graphs presented in Figure \ref{fig:NoCommonNeighborsWithTriangle}, and is therefore $3$-colorable (i.e., there is a homomorphism from $G$ to $K_3$). Since $G$ has an induced $K_3$, it proves that $core(G)=K_3$, contradicting that $G$ is a core.

\begin{center}
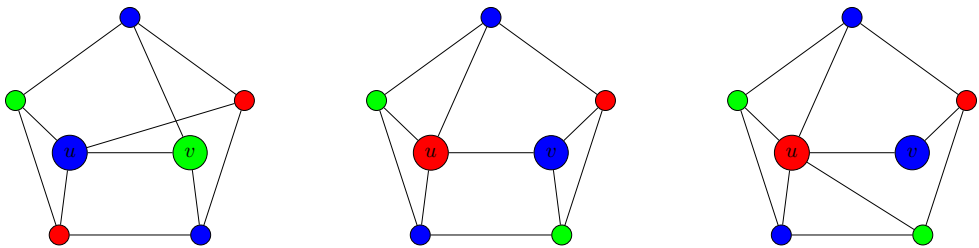
\begin{figure}
\centering

\scalebox{.8}{
\begin{tikzpicture}
\tikzstyle{vertex}=[draw,shape=circle];

\begin{scope}

    \node[vertex,fill=blue] (a) at (90:2) {};
    \node[vertex,fill=red] (e) at (-72+90:2) {}
        edge (a);
    \node[vertex,fill=blue] (d) at (-2*72+90:2) {}
        edge (e);
    \node[vertex,fill=red] (c) at (-3*72+90:2) {}
        edge (d);
    \node[vertex,fill=green] (b) at (-4*72+90:2) {}
        edge (a)
        edge (c);
    \node[vertex,fill=blue] (u) at (-1,-0.25) {$u$}
        edge (e)
        edge (c)
        edge (b);
    \node[vertex,fill=green] (v) at (1,-0.25) {$v$}
        edge (a)
        edge (d)
        edge (u);

\end{scope}

\begin{scope}[xshift=6cm]

    \node[vertex,fill=blue] (a) at (90:2) {};
    \node[vertex,fill=red] (e) at (-72+90:2) {}
        edge (a);
    \node[vertex,fill=green] (d) at (-2*72+90:2) {}
        edge (e);
    \node[vertex,fill=blue] (c) at (-3*72+90:2) {}
        edge (d);
    \node[vertex,fill=green] (b) at (-4*72+90:2) {}
        edge (a)
        edge (c);
    \node[vertex,fill=red] (u) at (-1,-0.25) {$u$}
        edge (b)
        edge (c)
        edge (a);
    \node[vertex,fill=blue] (v) at (1,-0.25) {$v$}
        edge (d)
        edge (e)
        edge (u);

\end{scope}

\begin{scope}[xshift=12cm]

    \node[vertex,fill=blue] (a) at (90:2) {};
    \node[vertex,fill=red] (e) at (-72+90:2) {}
        edge (a);
    \node[vertex,fill=green] (d) at (-2*72+90:2) {}
        edge (e);
    \node[vertex,fill=blue] (c) at (-3*72+90:2) {}
        edge (d);
    \node[vertex,fill=green] (b) at (-4*72+90:2) {}
        edge (a)
        edge (c);
    \node[vertex,fill=red] (u) at (-1,-0.25) {$u$}
        edge (b)
        edge (c)
        edge (a)
        edge (d);
    \node[vertex,fill=blue] (v) at (1,-0.25) {$v$}
        edge (e)
        edge (u);

\end{scope}

\end{tikzpicture}
}
\caption{Maximal case (up to isomorphism) where $u$ and $v$ do not have a common neighbor and  $G$ has a triangle. Even this maximal case is $3$-colorable.}
\label{fig:NoCommonNeighborsWithTriangle}
\end{figure}
\end{center}
    
\end{itemize}

In both cases, we have a contradiction. It follows that $u$ and $v$ have a common neighbor.

\end{proof}

\begin{lemma}\label{lem:NotAllCommonNeighbors}

There exists a vertex among $\{a,b,c,d,e\}$ that is not a common neighbor of $u$ and $v$.

\end{lemma}

\begin{proof}

We prove it by contradiction.
Note that if all vertices among $\{a,b,c,d,e\}$ are common neighbors of $u$ and $v$, then $G$ is one of the two graphs presented in Figure~\ref{fig:AllCommonNeighbors}.

\begin{center}
\begin{figure}
\centering

\begin{tikzpicture}
\tikzstyle{vertex}=[draw,shape=circle];

\begin{scope}

    \node[vertex] (a) at (90:2) {};
    \node[vertex] (e) at (-72+90:2) {}
        edge (a);
    \node[vertex] (d) at (-2*72+90:2) {}
        edge (e);
    \node[vertex] (c) at (-3*72+90:2) {}
        edge (d);
    \node[vertex] (b) at (-4*72+90:2) {}
        edge (a)
        edge (c);
    \node[vertex] (u) at (-1,-0.25) {$u$}
        edge (a)
        edge (b)
        edge (c)
        edge (d)
        edge (e);
    \node[vertex] (v) at (1,-0.25) {$v$}
        edge (a)
        edge (b)
        edge (c)
        edge (d)
        edge (e);

\end{scope}

\begin{scope}[xshift=7cm]

    \node[vertex] (a) at (90:2) {};
    \node[vertex] (e) at (-72+90:2) {}
        edge (a);
    \node[vertex] (d) at (-2*72+90:2) {}
        edge (e);
    \node[vertex] (c) at (-3*72+90:2) {}
        edge (d);
    \node[vertex] (b) at (-4*72+90:2) {}
        edge (a)
        edge (c);
    \node[vertex] (u) at (-1,-0.25) {$u$}
        edge (a)
        edge (b)
        edge (c)
        edge (d)
        edge (e);
    \node[vertex] (v) at (1,-0.25) {$v$}
        edge (a)
        edge (b)
        edge (c)
        edge (d)
        edge (e)
        edge (u);

\end{scope}

\end{tikzpicture}
\caption{The two possible graphs if all vertices in $\{a,b,c,d,e\}$ are joint neighbors of $u$ and $v$.}
\label{fig:AllCommonNeighbors}
\end{figure}
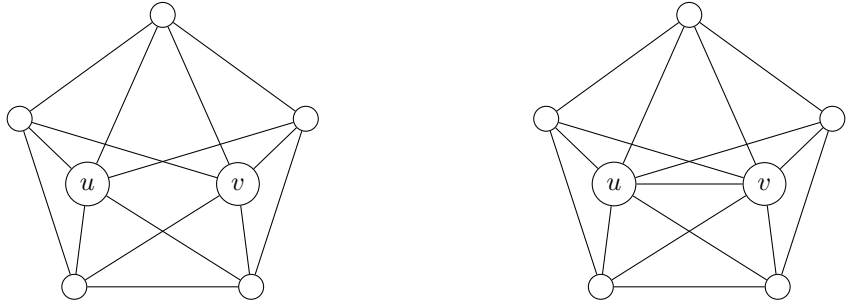
\end{center}

The first case is impossible, because the core of $G$ would then be $C_5+1$, and hence $G$ would not be a core. The second case is also impossible, because we assumed that $G$ is a sporadic $7$-core, so $G$ is different from $C_5+2$.
Clearly, there is a contradiction in both cases.

\end{proof}

We remark that by Lemma \ref{lem:CommonNeighbor}, $u$ and $v$ have a common neighbor, and we can thus assume (without loss of generality, up to isomorphism) that the vertex ``$a$'' is a common neighbor of both $u$ and $v$.

\begin{lemma}\label{lem:NoK4}

$G$ has no induced $K_4$.

\end{lemma}

\begin{proof}

Note that $G$ is a strict subgraph of $C_5+2$ obtained by removing at least an edge from one of the two universal vertices. We easily deduce that $G$ is $4$-colorable. It follows that $G$ does not contain a $K_4$, since otherwise $core(G)=K_4$.

\end{proof}

\begin{lemma}\label{lem:Triangle}

$G$ has a triangle.

\end{lemma}

\begin{proof}
With the goal of reaching a contradiction: if $G$ has no triangle, $G$ is isomorphic to a subgraph of the graph presented in Figure \ref{fig:NoTriangle}. Indeed, recall that $u$ and $v$ have a common neighbor by Lemma \ref{lem:CommonNeighbor}. Thus, there is a homomorphism from $G$ to $C_5$, and then $core(G)=C_5$, which is a contradicts that $G$ is a core.
\end{proof}
\begin{center}
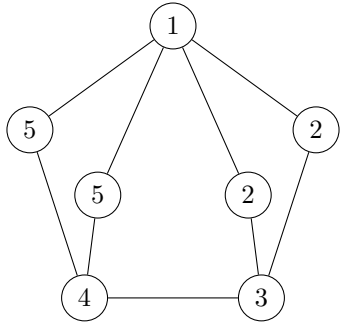
\begin{figure}
\centering

\begin{tikzpicture}
\tikzstyle{vertex}=[draw,shape=circle];

\begin{scope}

    \node[vertex] (a) at (90:2) {$1$};
    \node[vertex] (e) at (-72+90:2) {$2$}
        edge (a);
    \node[vertex] (d) at (-2*72+90:2) {$3$}
        edge (e);
    \node[vertex] (c) at (-3*72+90:2) {$4$}
        edge (d);
    \node[vertex] (b) at (-4*72+90:2) {$5$}
        edge (a)
        edge (c);
    \node[vertex] (u) at (-1,-0.25) {$5$}
        edge (a)
        edge (c);
    \node[vertex] (v) at (1,-0.25) {$2$}
        edge (a)
        edge (d);

\end{scope}

\end{tikzpicture}
\caption{Cases where $G$ has no triangle. $G$ is $C_5$-colorable.}
\label{fig:NoTriangle}
\end{figure}
\end{center}

\begin{corollary}\label{cor:Not3col}

$G$ is not $3$-colorable.

\end{corollary}

\begin{proof}

If $G$ was $3$-colorable, since, $G$ has a triangle by Lemma \ref{lem:Triangle}, $core(G)$ would be $K_3$.

\end{proof}

\begin{lemma}\label{lem:AtLeastOneNeighbor}

For every vertex $i$ in $\{a,b,c,d,e\}$, $i$ is a neighbor of $u$ or a neighbor of $v$.

\end{lemma}

\begin{proof}

Assume by contradiction that there exists a vertex $i$ in $\{a,b,c,d,e\}$ that is not a neighbor of $u$ and not a neighbor of $v$.
The $5$-cycle $\{a,b,c,d,e\}$ then becomes a $P_4=\alpha-\beta-\gamma-\delta$ when $i$ is removed, with $\alpha$ and $\delta$ being the two neighbors of $i$.
First, note that if the two neighbors $\alpha$ and $\delta$ of $i$ are also neighbors of $u$ (respectively $v$), then the function that maps $i$ to $u$ and that leaves the rest of the graph $G$ unchanged is a non-bijective homomorphism. This contradicts the fact that $G$ is a core.
We can now assume that the two neighbors $\alpha$ and $\delta$ of $i$ are not also two neighbors of $u$, nor are they two neighbors of $v$.

With this in mind we prove that $G$ is $3$-colorable. It is sufficient to establish that $G-i$ is $3$-colorable, because since $i$ has degree $2$, we will be able to extend this $3$-coloring to $G$ by coloring $i$ with (one of) the color(s) that is not taken by a neighbor of $i$.
We have the following two cases.

\begin{itemize}

\item If $u$ and $v$ are both neighbors of $\alpha$ or $\delta$: assume by symmetry that it is $\alpha$. Then, by what precedes, $\delta$ is not a neighbor of $u$ nor is it a neighbor of $v$. Since $\delta$ has only 1 neighbor among $\{\alpha,\beta,\gamma,u,v\}$ (it is $\gamma$), it is sufficient to 3-color $(G-i)-\delta$ to prove that $G-i$ is 3-colorable. Since $G$ has no induced $K_4$ by Lemma~\ref{lem:NoK4}, either $(u,v)\notin E_G$, $(\beta,u)\notin E_G$, or $(\beta,v)\notin E_G$, otherwise $\{\alpha,\beta,u,v\}$. In the first case, $(G-i)-\delta$ is 3-colorable by coloring $\alpha$ and $\gamma$ with the color 1, $u$ and $v$ with the color 2, and $\beta$ with the color 3. In the second case, color $\alpha$ and $\gamma$ with the color 1, $u$ and $\beta$ with the color $2$, and $v$ with the color $3$. In the third case, color $\alpha$ and $\gamma$ with the color 1, $v$ and $\beta$ with the color $2$, and $u$ with the color $3$.

The graph $(G-i)-\delta$ is 3-colorable, which implies that $G$ is 3-colorable, contradicting Corollary \ref{cor:Not3col}.

\item Since neither $u$ nor $v$ are common neighbors of $\alpha$ and $\delta$, and since neither $\alpha$ nor $\delta$ are common neighbors of $u$ and $v$, we can assume by symmetry that $(\alpha,v)\notin E_G$ and that $(\delta,u)\notin E_G$.

\begin{itemize}

    \item If $(u,v)\notin E_G$, color $u$ and $v$ with the same color, and $2$-color the rest of $G-i$ (which is the $P_4$ $\alpha-\beta-\gamma-\delta$).

    \item If $(u,v)\in E_G$, $u$ and $v$ can not have two adjacent common neighbors, otherwise $G$ would have a $K_4$ which is forbidden by Lemma \ref{lem:NoK4}. Thus, either $\beta$ or $\gamma$ is not a common neighbor of $u$ and $v$. We can assume by symmetry that $(\beta,v)\notin E_G$. $G-i$ is now $3$-colorable by coloring $\alpha$ and $\gamma$, $u$ and $\delta$, and $\beta$ and $v$ with the same color.

\end{itemize}

For an illustration of the two previous cases, $G-i$ is isomorphic to a subgraph of the two graphs presented in Figure \ref{fig:G-i}. 

\end{itemize}

\begin{center}
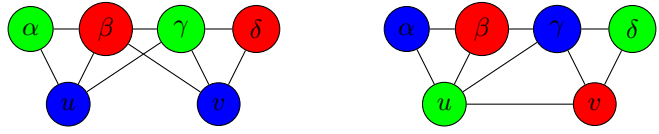
\begin{figure}
\centering

\begin{tikzpicture}
\tikzstyle{vertex}=[draw,shape=circle];

\begin{scope}

    \node[vertex,fill=green] (v1) at (-1.5,0) {$\alpha$};
    \node[vertex,fill=red] (v2) at (-0.5,0) {$\beta$}
        edge (v1);
    \node[vertex,fill=green] (v3) at (0.5,0) {$\gamma$}
        edge (v2);
    \node[vertex,fill=red] (v4) at (1.5,0) {$\delta$}
        edge (v3);
    \node[vertex,fill=blue] (u) at (-1,-1) {$u$}
        edge (v1)
        edge (v2)
        edge (v3);
    \node[vertex,fill=blue] (v) at (1,-1) {$v$}
        edge (v2)
        edge (v3)
        edge (v4);

\end{scope}

\begin{scope}[xshift=5cm]

    \node[vertex,fill=blue] (v1) at (-1.5,0) {$\alpha$};
    \node[vertex,fill=red] (v2) at (-0.5,0) {$\beta$}
        edge (v1);
    \node[vertex,fill=blue] (v3) at (0.5,0) {$\gamma$}
        edge (v2);
    \node[vertex,fill=green] (v4) at (1.5,0) {$\delta$}
        edge (v3);
    \node[vertex,fill=green] (u) at (-1,-1) {$u$}
        edge (v1)
        edge (v2)
        edge (v3);
    \node[vertex,fill=red] (v) at (1,-1) {$v$}
        edge (v3)
        edge (v4)
        edge (u);

\end{scope}

\end{tikzpicture}
\caption{$G-i$ is $3$-colorable}
\label{fig:G-i}
\end{figure}
\end{center}

Thus, $G$ is $3$-colorable, which contradicts Corollary~\ref{cor:Not3col}.
\end{proof}

\begin{remark}\label{rem:Defab}

By Lemma \ref{lem:CommonNeighbor}, $u$ and $v$ have a common neighbor among $\{a,b,c,d,e\}$, and by Lemma \ref{lem:NotAllCommonNeighbors}, not all vertices among $\{a,b,c,d,e\}$ are common neighbors of both $u$ and $v$. It is thus possible to find two adjacent vertices, say $a$ and $b$, such that $a$ is a common neighbor of $u$ and $v$, and such that $b$ is not a common neighbor of $u$ and $v$.
By Lemma \ref{lem:AtLeastOneNeighbor}, $b$ has at least one neighbor among $\{u,v\}$.
By symmetry between $u$ and $v$ we can assume that $b$ is a neighbor of $u$ but not a neighbor of $v$.

\end{remark}

\begin{lemma}\label{lem:Nodeg5}

$\text{deg}_G(u)<5$.

\end{lemma}

\begin{proof}

Assume, with the aim of reaching a contradiction, that $\text{deg}(u)\ge 5$.
If $(u,v)\notin E_G$, then $u$ is a neighbor of $a$, $b$, $c$, $d$ and $e$. The function that maps $v$ to $u$ and that leaves the rest of the graph unchanged is a non-bijective homomorphism. Contradiction with the fact that $G$ is a core.
Hence, assume now that $(u,v) \in E_G$.

\begin{itemize}
    \item 
If $\text{deg}_G(u)=5$, then there exists $i\in\{a,b,c,d,e\}$ such that $(u,i)\notin E_G$. The function that maps $i$ to $u$ and that leaves the rest of the graph unchanged is a non-bijective homomorphism. This contradicts that $G$ is a core.

\item
If $\text{deg}_G(u)>5$, then $G$ is a neighbor of $a,b,c,d$ and $e$. Since $G$ does not contain a $K_4$ by Lemma \ref{lem:NoK4}, the neighbors of $v$ in $\{a,b,c,d,e\}$ are non-adjacent. We deduce that the neighbors of $v$ in $\{a,b,c,d,e\}$ are contained in a set of the form $(\alpha,\beta)$, where $\alpha$ and $\beta$ are non-adjacent. $\alpha$ and $\beta$ have a common neighbor $\gamma$ in $\{a,b,c,d,e\}$. The function that maps $v$ to $\gamma$ and that leaves the rest of the graph unchanged is a non-bijective homomorphism. Again, this contradicts that $G$ is a core.
\end{itemize}

\end{proof}

\begin{lemma}\label{lem:evEdge}

$(e,v)\in E_G$.

\end{lemma}

\begin{proof}

Assume by contradiction that $(e,v)\notin E_G$. Then by Lemma \ref{lem:AtLeastOneNeighbor}, $(e,u)\in E_G$.
We have by definition of $b$ and $u$ that $(b,v)\notin E_G$ and $(b,u)\in E_G$.
By definition of $a$, $(a,u)\in E_G$ and $(a,v)\in E_G$.
By Lemma \ref{lem:Nodeg5}, $u$ has at least one non-neighbor among $\{a,b,c,d,e\}$. Since $u$ is a neighbor of $a$, $b$ and $e$, there are only two possible cases: either $(u,c)\notin E_G$ or $(u,d)\notin E_G$.

\begin{itemize}
    \item Assume that $(u,c)\notin E_G$. Coloring $b$, $v$ and $e$; $a$ and $d$; and $u$ and $c$ with the same color results in a $3$-coloring of $G$, which contradicts Corollary \ref{cor:Not3col}.
    \item Assume that $(u,d)\notin E_G$. Coloring $b$, $v$ and $e$; $a$ and $c$; and $u$ and $d$ with the same color results in a $3$-coloring of $G$, which contradicts Corollary \ref{cor:Not3col}.
\end{itemize}

In either case there is a contradiction. Thus, $(e,v)\in E_G$.
\end{proof}

For each vertex let us now summarizes the remaining possible cases.

\begin{itemize}
    \item $a$: 0 choices: $(u,a)$ and $(v,a)$ are edges of $G$ by Remark \ref{rem:Defab}.
    \item $b$: 0 choices: $(u,b)$ is an edge, and $(v,b)$ is not an edge by Remark \ref{rem:Defab}.
    \item $c$: 3 choices since $c$ must be either a neighbor of $u$ or of $v$ by Lemma \ref{lem:AtLeastOneNeighbor}:
    \begin{itemize}
        \item $(u,c)$ and $(v,c)$ are edges of $G$.
        \item $(u,c)$ is an edge, and $(v,c)$ is not an edge.
        \item $(v,c)$ is an edge, and $(u,c)$ is not an edge.
    \end{itemize}
    \item $d$: 3 choices since $d$ must be either a neighbor of $u$ or of $v$ by Lemma \ref{lem:AtLeastOneNeighbor}:
    \begin{itemize}
        \item $(u,d)$ and $(v,d)$ are edges of $G$.
        \item $(u,d)$ is an edge, and $(v,d)$ is not an edge.
        \item $(v,d)$ is an edge, and $(u,d)$ is not an edge.
    \end{itemize}
    \item Concerning $u$, $v$ and $e$, since $(e,v)$ in an edge by Lemma \ref{lem:evEdge}, there are 3 possible cases:
    \begin{itemize}
        \item Neither $(u,v)$ nor $(e,u)$ is an edge.
        \item $(u,v)$ is an edge and $(e,u)$ is not an edge.
        \item $(e,u)$ is an edge and $(u,v)$ is not an edge.
    \end{itemize}
    The case where both $(u,v)$ and $(e,u)$ is an edge is impossible, because otherwise $\{a,e,u,v\}$ would induce a $K_4$, contradicting Lemma \ref{lem:NoK4}.
\end{itemize}

We conclude that, up to isomorphism, all the sporadic $7$-cores belong to the list of the $3\times 3\times 3=27$ graphs presented in Figures \ref{fig:NoteuNotuv}, \ref{fig:uvNoteu} and \ref{fig:euNotuv}.

\begin{corollary}
All the sporadic $7$-cores are contained in the graphs allowed in Figure~\ref{fig:AllPossibleGraphs}.
\end{corollary}

\end{document}